\documentclass[reprint,aps,pra,amsmath,amssymb,longbibliography,floatfix]{revtex4-2}

\usepackage{graphicx}
\usepackage{hyperref}
\usepackage{physics}
\usepackage{amsthm}
\usepackage{booktabs}
\usepackage{enumitem}

\newtheorem{definition}{Definition}
\newtheorem{theorem}{Theorem}
\newtheorem{proposition}{Proposition}
\newtheorem{lemma}{Lemma}
\newtheorem{corollary}{Corollary}
\newtheorem{conjecture}{Conjecture}
\theoremstyle{remark}
\newtheorem{remark}{Remark}

\newcommand{\pin}{p_{\mathrm{in}}}
\newcommand{\plam}{p_{\lambda}}
\newcommand{\Lset}{\mathcal{L}}
\newcommand{\sL}{s^{L}}

\begin{document}

% \title{Verifying a remote quantum processor with degenerate concordant computations,
% with an application to randomness generation}

\title{A Framework for Classical Verification of a Remote Quantum Processor via Degenerate Concordant Circuits, with Application to Randomness Generation}

\author{Bishal Kumar Das}
\email{bishal.k@qpiai.tech}
\affiliation{QpiAI India Private Limited, Bengaluru, Karnataka, India}

\author{Lakshya Priyadarshi}
% \email{lakshya.priyadarshi@qpiai.tech}
\affiliation{QpiAI India Private Limited, Bengaluru, Karnataka, India}

\date{\today}

\begin{abstract}
A client who rents a cloud quantum processor cannot easily tell whether the
bitstrings it receives came from a quantum device. Certification by random circuit
sampling runs on present hardware, but scoring each returned sample costs the client
time $O(2^n)$. We give a protocol in which the client scores each round in time
$O(n)$. It is built on degenerate concordant ensembles: mixed states that are
classical probability distributions over a product basis, in which several basis
states share the same probability. For such a state, every gate that preserves
concordance factorizes into local basis changes, a relabeling of basis states, and a
block factor that rotates among labels of equal probability. The block factor commutes
with the state, so it cancels from the ensemble while entangling each pure component. We call circuits of this kind degenerate concordant circuits (DCCs). The client
holds a secret parity whose cosets are the degenerate sectors of every
layer, so that it builds the circuit and predicts the ensemble with one inner
product per layer. The secret never leaves the client. Inputs are sent uniformly
at random, and the client attaches to each round, in its own bookkeeping, the
eigenvalue of that input in the hidden state; both verification statistics, an
input--output agreement test and a weighted collision test, then compare the
server's samples with the non-flat output distribution of the concordant
computation. We prove completeness, client efficiency, and soundness in an ideal
model in which the compiled circuit reveals nothing about the secret, and we show
that the secret must stay hidden: a classical server that holds it passes. For the
explicit construction we prove that the single-qubit symmetry test used in efficient
simulations of concordant computation fails on non-product inputs with unbiased
marginals, and we measure that a published local-basis finder halts on
every layer we tested, that the transmitted component saturates the maximal Schmidt
rank, that the circuit is not Clifford, and that matrix-product spoofers are
rejected. No known classical algorithm in the relevant literature recovers the secret
of the explicit construction from the public data. Other choices of blocks and
permutations are possible, and the client may build the circuit from any of them.
Soundness is stated as a conjecture for constructions whose public data carry no
signature of the secret.
\end{abstract}

\maketitle

% ==================================================================
\section{Introduction}\label{sec:intro}
% ==================================================================

A client rents time on a remote quantum computer and receives classical
bitstrings. It has no direct way to know whether a quantum device produced them.
For random number generation this is the whole difficulty, because the value of
the output rests entirely on the producer not having been able to predict it.
Classical pseudo-random generators unravel once the seed is known, and physical
noise sources can be manipulated from outside~\cite{Acin2016}. Quantum random
number generators tie unpredictability to the Born
rule~\cite{HerreroCollantes2017}, but a generator one does not own is only as
good as the argument that it ran.

Checking a quantum computation from the outside is a problem with a large
literature, reviewed in Ref.~\cite{Gheorghiu2019}. Protocols in which the client
prepares or measures single qubits, as in verifiable blind
computation~\cite{FitzsimonsKashefi2017}, give strong guarantees but need a
quantum channel to the client. For a purely classical client three families of
answers exist, each expensive in a different currency.
\emph{Device-independent} generation~\cite{Colbeck2009,Pironio2010} certifies
through violation of a Bell inequality~\cite{Brunner2014} and assumes very little
about the hardware, but needs the measured subsystems to be space-like separated.
On a single chip, gate times are nanoseconds and the light-crossing time between
qubits is picoseconds~\cite{Krantz2019}, so that separation does not exist.
\emph{Cryptographic verification} removes the separation
requirement~\cite{Mahadev2018,BCMVV2018}, and has been refined to reduce the quantum
resources it needs~\cite{KahanamokuMeyer2022}, at the price of depth: small
instances have been implemented on trapped ions~\cite{Zhu2023,Lewis2024}, but at
sizes where the underlying cryptographic problem is hard, evaluating a trapdoor
claw-free function coherently needs circuits deeper than present hardware
supports~\cite{Preskill2018}. \emph{Random circuit sampling} is the task behind the
quantum-advantage experiments~\cite{Arute2019,Zhong2020}; its hardness rests on
complexity-theoretic arguments~\cite{Aaronson2011,Bouland2019,HangleiterEisert2023},
and samples are scored by linear cross-entropy benchmarking~\cite{Boixo2018}.
Certification by random circuit sampling~\cite{AaronsonHung2023} runs on present
hardware and has been demonstrated on a 56-qubit trapped-ion
processor~\cite{Liu2025}, but the client must compute an output probability of a
pseudo-random circuit to score each sample. This exponential cost is identified in
Ref.~\cite{AaronsonHung2023} as the central drawback of the protocol, one that in
practice limits it to about sixty qubits.

Verification has also been where these schemes needed the most care. The linear
cross-entropy score has been reproduced by classical algorithms that exploit its
structure rather than simulate the circuit~\cite{Gao2024}; tensor-network methods
have sampled the original advantage circuits at the reported
fidelity~\cite{PanChenZhang2022}; and noisy random circuits in the
anticoncentration regime admit polynomial-time classical
sampling~\cite{Aharonov2023}. Spoofing the score has been related to a natural
complexity conjecture~\cite{AaronsonGunn2020}, and certifying the output
distribution directly is known to require a number of samples growing
exponentially with the register size~\cite{Hangleiter2019}. The common lesson is
that a verification statistic must be judged by what it cannot see.

A fourth approach plants a secret in the circuit itself, and its history is the
closest precedent for this paper. Circuits of commuting gates, the class IQP, are hard
to sample classically under plausible conjectures, exactly, up to additive error, and
with sparse noise~\cite{BremnerJozsaShepherd2011,BremnerMontanaroShepherd2016,
BremnerMontanaroShepherd2017}. The first protocol of this kind hides a secret binary
vector in such a circuit, so that samples from an honest device are correlated with it
and a client who knows it can check them efficiently~\cite{ShepherdBremner2009}. The
secret was later shown to be extractable from the published circuit alone, by a
classical algorithm efficient for hundreds of qubits, which solved the challenge
instance posted with the protocol~\cite{KahanamokuMeyer2019}. Revised constructions
followed. The protocol was extended so as to avoid that attack~\cite{YungCheng2020},
and the test has been run on a five-qubit nuclear-magnetic-resonance processor, which
passed it, and on cloud superconducting processors, whose fidelity was then too low to
pass~\cite{ChenCheng2021}; the construction has also been generalized to a stabilizer
scheme whose classical security was analyzed in Ref.~\cite{BremnerChengJi2025}.
Secret-extraction attacks effective against both the original and the generalized
scheme, over a wide range of parameters, have since been
developed~\cite{GrossHangleiter2025}. A related line hides
a single large peak in the output distribution of a random-looking
circuit~\cite{AaronsonZhang2024,Deshpande2025,Gharibyan2025}. Whether a secret can be
planted so that classical verification is efficient while extraction stays hard is
still open. A protocol of this kind stands or falls on whether the secret can be
extracted from what the server is given, and this paper is organized around that
question.

We take the secret from a different place: the degeneracy structure of a
concordant computation. A mixed state is concordant if it is a classical
probability distribution over the states of a single product basis, so that it
carries neither entanglement nor quantum discord; a computation is concordant if
its state remains concordant after every gate, although the product basis may
change from one gate to the next~\cite{Eastin2010}. Such computations are natural
candidates for classical simulation, and they have been simulated efficiently in
several settings~\cite{Eastin2010,Cable_2015}. When several basis states carry the
same probability the state is degenerate, and degeneracy introduces freedom in the
product basis used to describe it: rotations inside a degenerate sector are
invisible to the ensemble but nontrivial on its individual components. We exploit
this freedom by giving the client a compact secret description of the degenerate
sectors. The client can therefore construct and predict the concordant ensemble
efficiently, whereas the server receives only the compiled circuit and must infer
the hidden sector structure from it.

This paper makes these ideas precise. Section~\ref{sec:background} introduces
concordant states and concordant computation and reviews what is known about
simulating them classically. Section~\ref{sec:framework} shows how degeneracy splits
every gate into a part that the ensemble registers and a part that only individual
runs feel. Section~\ref{sec:construction} gives the construction. Once the client has
chosen the secret parity, the sectors of every layer follow from a single inner
product, so the client builds the whole circuit in time that grows only linearly with
the number of layers and never stores anything of size $2^n$
(Theorem~\ref{thm:poly}). The gates are chosen so that the circuit is not Clifford and
so that no single compiled gate reveals the sectors. Section~\ref{sec:protocol} gives
the protocol. The client sends uniformly random inputs, attaches to each round the
eigenvalue of its input in the hidden state (Proposition~\ref{prop:weights}), and
checks the returned bitstrings with two tests, an agreement test and a weighted
collision test. Each test costs $O(n)$ per round, and for both the client knows exactly
what an honest device should produce (Propositions~\ref{prop:agree}
and~\ref{prop:collision}).

Section~\ref{sec:soundness} turns to soundness. A server that holds the secret passes
both tests (Proposition~\ref{prop:trapdoor}), so the secret must stay hidden; a server
that learns nothing about it cannot pass the agreement test, except with small
probability (Theorem~\ref{thm:ideal}). We then test the explicit circuit against the
known ways of simulating concordant computations. We prove that the single-qubit
symmetry test used in earlier efficient simulations fails on our input
(Theorem~\ref{thm:symtest}); the local-basis finder of Ref.~\cite{Cable_2015} halts on
every layer we tested; the state of each run reaches maximal entanglement across the
middle of the register on every eight-qubit instance; and truncated matrix-product
simulations are rejected. No known classical algorithm in the relevant
literature~\cite{Eastin2010,Cable_2015} recovers the secret, or the frames and sectors
that would reveal it, from the public data of the explicit circuit. The construction is
one choice among several: other blocks and permutations that keep the client's
record-keeping polynomial are possible, and the client may build the circuit from any
of them. Section~\ref{sec:instances} reports these measurements, on
instances of up to ten qubits, Section~\ref{sec:entropy} the randomness that a passing
run certifies, and Section~\ref{sec:discussion} the number of rounds the protocol needs
and how it might be reduced. Completeness, the client's cost, the properties of both
tests, Theorem~\ref{thm:symtest} and soundness in the ideal model are proved. Soundness
against classical servers running in polynomial time is stated as
Conjecture~\ref{conj:main}, for constructions whose public data carry no signature of
the secret. All security statements concern classical servers.

% ==================================================================
\section{Concordant states and concordant computation}\label{sec:background}
% ==================================================================

Concordant computation is less familiar than the pure-state circuit model, and the
protocol of this paper rests on features of it that are rarely spelled out. This
section introduces it from the beginning, explains why it has been studied, and
reviews what is known about simulating it on a classical computer. Readers familiar
with Refs.~\cite{Eastin2010,Cable_2015} may skip to Sec.~\ref{sec:framework}.
Appendix~\ref{app:notation} collects the symbols used in the rest of the paper.

\paragraph*{Conventions.} The register has $n$ qubits. Its computational-basis
states $\ket x$ are labeled by bit strings $x$ of length $n$, and there are $D=2^n$
of them. Bit strings are added bitwise modulo two, written $\oplus$, that is, treated
as vectors over the two-element field $\mathbb F_2$; $|s|$ is the number of ones in a
string $s$, its Hamming weight. For two strings
$s$ and $x$, $\langle s,x\rangle$ is the parity of the bits of $x$ at the positions
where $s$ has a one. A parity has a direct physical meaning: $(-1)^{\langle
s,x\rangle}$ is the eigenvalue on $\ket x$ of the product $Z^s$ of Pauli $Z$ operators
on those positions, so a parity is the value of a many-body $Z$ correlation, written as
a bit.

\subsection{Concordant states}\label{sec:concordance}

A mixed state can be written as a statistical mixture of pure states in many ways,
and the states in such a mixture need not be orthogonal. A concordant state is one for
which one of these mixtures is a classical probability distribution over the states of
a single product basis.

\begin{definition}[Concordant state]\label{def:conc}
A state $\rho$ on $N$ subsystems is concordant if
\begin{equation}
\rho=\sum_{k_1,\dots,k_N}p(k_1,\dots,k_N)\,
\pi^{(1)}_{k_1}\otimes\cdots\otimes\pi^{(N)}_{k_N},
\label{eq:concdef}
\end{equation}
where for each subsystem $j$ the $\pi^{(j)}_{k_j}$ are orthogonal rank-one projectors
and $p$ is a probability distribution.
\end{definition}

In words, each subsystem $j$ carries its own orthonormal basis, the vectors onto which
the projectors $\pi^{(j)}_{k}$ project, and $\rho$ is diagonal in the product basis
built from these local bases. The label $(k_1,\dots,k_N)$ of a product basis state
behaves as a classical random variable with distribution $p$: measuring every
subsystem in its own basis returns the label with probability $p(k_1,\dots,k_N)$ and,
averaged over the outcomes, leaves the state unchanged. Three two-qubit examples fix
the idea.
\begin{enumerate}[label=(\roman*),leftmargin=*,itemsep=1pt]
\item A product state $\rho_1\otimes\rho_2$ is concordant. Each factor is diagonal in
its own eigenbasis, and $p$ is the product of the two spectra.
\item The state $(\ket{00}\!\bra{00}+\ket{11}\!\bra{11})/2$ is concordant but not a
product. Its qubits behave as two classical bits that always agree, and $p$ does not
factorize.
\item The state $(\ket{00}\!\bra{00}+\ket{{+}{+}}\!\bra{{+}{+}})/2$, with
$\ket{+}=(\ket0+\ket1)/\sqrt2$, is a mixture of product states and therefore
unentangled, but it is not concordant: the two product states are not orthogonal, and
no product basis diagonalizes the mixture.
\end{enumerate}
Example (iii) shows that concordance is stronger than the absence of entanglement. The
difference is quantified by quantum discord~\cite{Ollivier2001,Henderson2001,Modi2012},
which measures how much of the correlation between two parts of a state is necessarily
disturbed when one part is measured locally. A state is concordant exactly when its
discord vanishes across every division of the register into two
parts~\cite{Eastin2010,Modi2012}, so concordant states are in this sense the most
classical of mixed states. They need not be simple, however. The state used in this
paper (Sec.~\ref{sec:hiddenstate}) is of type (ii): diagonal in the computational
basis, with a distribution $p$ that is correlated across many qubits.

One feature of concordant states will be central. The product basis is not always
fixed by the state. If two basis states carry the same probability, any rotation
inside the space they span leaves the state unchanged, just as the maximally mixed
state of a qubit is diagonal in every basis. Such equal probabilities are
\emph{degeneracies}, and Sec.~\ref{sec:framework} shows that they hide information
from anyone who holds only the ensemble.

\subsection{Concordant computation}\label{sec:conccomp}

\begin{definition}[Concordant computation]\label{def:conccomp}
A quantum computation with input $\rho_0$ and gates $G_1,\ldots,G_L$ is concordant if
$\rho_0$ and every intermediate state
\begin{equation*}
\rho_t=G_t\cdots G_1\rho_0G_1^\dagger\cdots G_t^\dagger
\end{equation*}
are concordant~\cite{Eastin2010}.
\end{definition}

The state therefore remains diagonal in some product basis throughout the computation,
but not necessarily in the same one: each gate may rotate the product basis in which
the state is diagonal. This requirement strongly limits what a gate can do. If no two
basis states carry the same probability, the eigenvectors of the state are fixed, and a
gate that keeps the state diagonal in a product basis must carry each state of the old
product basis onto a state of the new one. Seen in these two bases, the gate is a
permutation of labels up to phases, and the whole computation is a reversible
classical computation on labels, dressed in changing local bases. Degeneracy relaxes
this: labels that carry the same probability may be mixed with one another without
changing the state. Section~\ref{sec:framework} makes this precise.

Concordant computation was introduced to answer a question about the resources behind
quantum speedups~\cite{Eastin2010}. In the one-clean-qubit model, a single pure qubit
and many maximally mixed ones suffice to estimate the normalized trace of a unitary, a
task with no known efficient classical algorithm, while the entanglement present stays
small~\cite{Knill1998,Datta2008}. Quantum discord is present in that
computation~\cite{Datta2008}, which suggested discord, rather than entanglement, as the
resource behind mixed-state speedups and raised the converse question: is a
computation that never generates discord easy to simulate
classically~\cite{Eastin2010}? The role of correlations beyond entanglement in
computation has since been examined in detail~\cite{Modi2012,Brodutch2013}. Discord
also has operational meaning outside computation. It is the markup in the cost of
quantum state merging when relevant prior information is
discarded~\cite{MadhokDatta2011PRA}, it quantifies the advantage of quantum coherence
in noisy communication protocols~\cite{MadhokDatta2014TQC,MadhokDatta2013IJMPB,
DattaMadhok2017Lectures}, and its dynamics carry signatures of chaos and
ergodicity~\cite{Madhok2015PRE,Madhok2018OptComm}.

\subsection{Classical simulation of concordant computation}\label{sec:concsim}

Because a concordant state is, in the right basis, only a probability distribution
over labels, one may hope to simulate a concordant computation by following labels:
draw a label from the input distribution and update it gate by gate, as the
computation permutes labels. Two difficulties stand in the way. The simulator must know
the product basis at every step, and it must know which labels are degenerate, because
a gate may mix degenerate labels instead of permuting them. The literature has
resolved these difficulties in progressively more general settings, summarized in
Table~\ref{tab:known}.

The first efficient simulation~\cite{Eastin2010} treats a simplified form of this
structure. The input is a product state diagonal in a known basis, the special case
(i) of Sec.~\ref{sec:concordance} in which $p$ factorizes and every local basis is
given, and the gates act on one or two qubits. In that setting the permutations a gate
can perform are affine maps of the label bits, and deciding which labels are
degenerate reduces to $n+1$ evaluations. Nondegenerate product inputs are tractable
for larger gates as well, because there is then no degeneracy to determine. With
degeneracy and gates on three or more qubits the situation changes: even a
transposition of two labels is then not affine, and testing whether an arbitrary
permutation leaves a degenerate state unchanged embeds 3-SAT~\cite{Eastin2010}. The
conclusion drawn in Ref.~\cite{Eastin2010} is that a conventional computation of one-
and two-qubit gates on a product input must generate discord if it is to solve
efficiently a problem that is hard classically.

Reference~\cite{Cable_2015} reformulated the problem in terms of the local product
basis itself, extended it to larger gates and to qudits, and argued that a larger
class of concordant computations admits efficient simulation. Its local-basis finder
(LBF) propagates a concordant state gate by gate by searching for local rank-one
projectors compatible with the transformed subsystem eigenspaces, the operators
introduced in Sec.~\ref{sec:sectors}. When the candidate projectors determine
compatible local bases, the evolution continues efficiently for gates of fixed size.
When different admissible assignments of degenerate labels imply incompatible local
bases, the procedure cannot select a unique continuation and halts. Whether every
concordant computation can be simulated efficiently is left open~\cite{Cable_2015}.

\begin{table}[t]
\caption{Known results on simulating concordant computation classically, and the
setting of this paper. ``Known input'' means that the simulator is given the input
state.}
\label{tab:known}
\begin{ruledtabular}
\begin{tabular}{p{0.5\columnwidth}p{0.42\columnwidth}}
setting & classical simulation\\
\hline
product input in a known basis; one- and two-qubit gates & efficient~\cite{Eastin2010}\\
nondegenerate product input; larger gates & efficient~\cite{Eastin2010}\\
degenerate input; gates on three or more qubits & degeneracy test for an arbitrary
permutation embeds 3-SAT~\cite{Eastin2010}\\
known input; the LBF fixes the local bases at every gate & efficient for gates of
fixed size~\cite{Cable_2015}\\
known input; incompatible local bases & the LBF halts; general case
open~\cite{Cable_2015}\\
this paper: correlated, degenerate input whose sectors are withheld from the
simulator & the sectors must be inferred from the compiled gates
(Sec.~\ref{sec:soundness})\\
\end{tabular}
\end{ruledtabular}
\end{table}

Both lines of work point to the same structural issue: \emph{degeneracy}. Without
degeneracy the spectrum fixes the eigenvectors, and the product basis and the
permutation can be read off each gate. With degeneracy there is freedom to rotate
inside equal-probability subspaces without changing the mixed state, and resolving
that freedom is the hard part of the simulation.

% ==================================================================
\section{Theoretical framework: degeneracy as a hidden resource}\label{sec:framework}
% ==================================================================

This section makes the role of degeneracy precise. Every gate of a concordant
computation splits into a part that the mixed state registers, a relabeling, and a
part that it does not, a rotation among degenerate labels. The second part still acts
on every individual run. This asymmetry is what the protocol exploits.

\subsection{Frames, supports and sectors}\label{sec:sectors}

We consider a circuit of $L$ layers. Layer $t$ applies a unitary $G_t$ to a small set
$b_t$ of $\ell$ qubits, its \emph{support}, which carries $d_b=2^\ell$ labels; $a$
denotes the remaining qubits, and a label splits into a part $\vec x$ on the support
and a part $v$ on the rest. The input is diagonal in the computational basis,
$\rho_0=\sum_x\pin(x)\ket x\!\bra x$, with a probability distribution $\pin$ over
labels. The \emph{level set} $\Lset(x)$ of a label is the set of labels that have the
same input probability as $x$.

After $t$ layers the state $\rho_t$ is concordant, so it is diagonal in some product
basis. We describe that basis by the product of single-qubit unitaries $U_t$ that
rotates the computational basis onto it, and call $U_t$ the \emph{frame}. The input is
prepared in the computational basis, so $U_0=I$. The state seen from its own frame,
\begin{equation}
\sigma_t=U_t^\dagger\rho_t\,U_t ,
\label{eq:sigmat}
\end{equation}
is diagonal in the computational basis: it is the table of probabilities of the labels
after $t$ layers, with $\sigma_0=\rho_0$.

Degeneracy must be judged relative to the support of the acting gate, because the
gate touches only its support while the rest of the register may still tell labels
apart. Fix the label $v$ of the qubits outside the support. What remains of the
diagonal state is a list of joint probabilities, one for each support label $\vec x$.
Collecting these for all $v$, the state can be written as $\rho=\sum_{\vec x}\tilde
\rho^{(a)}_{\vec x}\otimes\pi^{(b)}_{\vec x}$, where $\pi^{(b)}_{\vec x}$ projects the
support onto label $\vec x$ and the unnormalized \emph{conditional state}
$\tilde\rho^{(a)}_{\vec x}$ describes the rest of the register given that
label~\cite{Cable_2015}. When the support is the whole register the conditional state
is a single number, the probability of $\vec x$; in general it plays the role that an
eigenvalue plays then. Two support labels are \emph{degenerate} if their conditional
states coincide, so that nothing in the rest of the register distinguishes them.
Grouping degenerate labels gives the unique decomposition
\begin{equation}
\rho=\sum_i\tilde\rho^{(a)}_i\otimes\Pi^{(b)}_i,\qquad
\Pi^{(b)}_i=\sum_{\vec x\in\mathcal S_i}\pi^{(b)}_{\vec x},
\label{eq:sectors}
\end{equation}
whose sets $\mathcal S_i$ we call the degenerate \emph{sectors}; the $\Pi^{(b)}_i$ are
the full-rank subsystem eigenprojectors of Ref.~\cite{Cable_2015}. Within a sector the
support is maximally mixed, so any unitary on the support that acts only inside the
sectors leaves $\rho$ unchanged.

\subsection{Factorization of a concordance-preserving gate}\label{sec:cancellation}

The significance of this freedom is not captured by discord alone.
Reference~\cite{das2026objectivitybasisambiguityclassical} emphasized that for degenerate concordant states the local
product basis needed for a classical description need not be fixed objectively by the
state itself. The missing local-basis information can therefore remain nontrivial even
though the ensemble has zero discord. Degeneracy provides this freedom, but is not in
itself a hardness guarantee. In the present work we use this basis freedom
operationally: the client knows the degenerate sectors from which the relevant basis
information can be propagated, while the server is not given that structure.

\begin{proposition}[Factorization~\cite{Eastin2010,Cable_2015,das2026objectivitybasisambiguityclassical}]\label{prop:factor}
Suppose the state before layer $t$ is diagonal in the frame $U_{t-1}$, and the gate
$G_t$ leaves it diagonal in a frame $U_t$ that agrees with $U_{t-1}$ outside $b_t$.
Then
\begin{equation}
G_t=U_tP_tB_tU_{t-1}^\dagger ,
\label{eq:gate}
\end{equation}
where $P_t$ permutes the labels of the support and the \emph{block factor} $B_t$ acts
inside the sectors of $\sigma_{t-1}$ on $b_t$, so that $[B_t,\sigma_{t-1}]=0$.
\end{proposition}

Read from right to left, Eq.~\eqref{eq:gate} says what the gate does.
$U_{t-1}^\dagger$ rotates the support from the old frame back to the computational
basis, where the state is a table of label probabilities; $B_t$ rotates inside each
sector, which the table does not register; $P_t$ relabels; and $U_t$ rotates into the
new frame. Only $P_t$ changes the table. We write $C=G_L\cdots G_1$ for the circuit and
$P=P_L\cdots P_1$ for its composite permutation.

\begin{theorem}[Ensemble invariance]\label{thm:invariance}
Under Proposition~\ref{prop:factor},
$\sigma_t=P_tB_t\sigma_{t-1}B_t^\dagger P_t^\dagger=P_t\sigma_{t-1}P_t^\dagger$, so
measuring $C\rho_0C^\dagger$ in the frame $U_L$ gives $y$ with probability
$\pin(P^{-1}(y))$.
\end{theorem}

The block factor drops out for the same reason that rotating a maximally mixed qubit
does nothing: inside a sector the support is maximally mixed. The ensemble therefore
evolves as a classical relabeling, and the final distribution is the input
distribution carried through the composite permutation.

\begin{corollary}[Branch confinement]\label{cor:confine}
For every label $x$, $U_L^\dagger C\ket x$ is supported on the level set
$P(\Lset(x))$.
\end{corollary}

Proofs are in Appendix~\ref{app:factor}. A single run starts from one label, not from
the mixture, and the corollary says where it can go: never outside the image of its
own level set. Inside that set it spreads. On a single run the frames cancel between
consecutive layers, $U_L^\dagger C=P_LB_L\cdots P_1B_1$, so that
\begin{equation}
U_L^\dagger C\ket x=\!\!\sum_{y\in P(\Lset(x))}\!\!\alpha_y(x)\ket y .
\label{eq:component}
\end{equation}
Theorem~\ref{thm:invariance} and Eq.~\eqref{eq:component} carry the physics of the
paper. The blocks are absent from the first and present in the second: the
permutations decide \emph{which} labels can carry weight, the blocks decide \emph{how
much} each one gets. A classical description of the ensemble needs only the
permutations, whereas a device that runs the circuit produces individual runs, whose
amplitudes $\alpha_y(x)$ depend on every block. The same distinction between an
ensemble description and the record of an individual run separates the decoherence
account of the quantum-to-classical transition from a trajectory-level
one~\cite{Arora2025}.

Three regimes are worth naming. If no two support labels are degenerate, $B_t$ is
diagonal, the gate is a permutation up to phases in the correct frame, and nothing is
hidden. If some labels are degenerate, $B_t$ is a genuine rotation inside each sector:
it leaves the ensemble untouched while entangling the component. If all support labels
are degenerate, $B_t$ may be any unitary on the support and absorbs the permutation.
The construction stays in the middle regime.

\begin{definition}[Degenerate concordant circuit]\label{def:dcc}
Let $\rho_0$ be a degenerate concordant state. A circuit $C=G_L\cdots G_1$ is a
\emph{degenerate concordant circuit} (DCC) for $\rho_0$ if every intermediate state
$\rho_t$ is concordant and, in the factorization of Proposition~\ref{prop:factor},
no block factor $B_t$ is diagonal, so that every layer mixes labels within a sector.
\end{definition}

By Theorem~\ref{thm:invariance} a DCC acts on $\rho_0$ as a relabeling, while by
Eq.~\eqref{eq:component} it entangles every individual run. A DCC is concordant only
on $\rho_0$: started from a single label, as in every round of the protocol, its
runs are pure entangled states.

\subsection{From hidden degeneracy to efficient verification}
\label{sec:verificationidea}

A simulator faces an inverse problem: given the compiled gates, recover enough of the
factorization to propagate a description of the state. A client that had to solve the
same problem forward would inherit this cost. The client here never diagnoses
degeneracy: it manufactures it and carries an exact description of it forward.

The verification idea is to make this asymmetry deliberate. The client builds the
circuit around a secret description of the degenerate sectors, so it can propagate the
concordant ensemble without solving the inverse problem faced by a
simulator~\cite{Eastin2010,Cable_2015}. In the explicit construction below that
description is a parity $s$: its two cosets specify the sectors at every layer and can
be updated using only polynomial-size classical data.

At the same time, the block factors $B_t$ cancel from the ensemble by
Theorem~\ref{thm:invariance} while remaining active on each individual component, as
shown in Eq.~\eqref{eq:component}. The client can therefore evaluate the ensemble using
a compact classical description even though an individual run experiences the full
block dynamics. This is the asymmetry on which the verification protocol is built;
whether a concrete construction is sound is the conjecture stated in
Sec.~\ref{sec:soundness}.

The hidden sector description must also stay out of the inputs. A classical client has
to tell the server which basis state to prepare in each round; if those inputs were
drawn directly from the concordant state, the transmitted labels would themselves
contain information about its degeneracy structure. The protocol therefore sends
uniformly random inputs and restores the statistics of the hidden state only in the
client's private bookkeeping.

The next section gives an explicit family for which the hidden sector description is a
single parity and can be propagated efficiently.

% ==================================================================
\section{The construction}\label{sec:construction}
% ==================================================================

\subsection{The hidden state}\label{sec:hiddenstate}

The circuits we construct are DCCs for the hidden-parity state of Eq.~\eqref{eq:rho0}; the recipe of Appendix~\ref{app:construction} is one way to build them. The input of the concordant computation is never sent to the server. It is fixed by a
secret string $s$ with at least three ones, $|s|\ge3$, and a bias $\lambda\in(0,1]$:
\begin{equation}
\rho_0(s,\lambda)=\sum_x\pin(x)\ket x\!\bra x=\frac{I+\lambda Z^{s}}{D}.
\label{eq:rho0}
\end{equation}
It is the maximally mixed state plus a single correlation among the $|s|$ qubits
selected by $s$: the product of their $Z$ values is biased towards $+1$, and nothing
else is. The state is diagonal in the computational basis, hence concordant, and its
distribution $\pin$ is correlated across those qubits, as in example (ii) of
Sec.~\ref{sec:concordance}. The value of the correlation on a label, written as a bit,
is the \emph{class} of the label, $c(x)=\langle s,x\rangle$. A label of class $0$ has
probability $(1+\lambda)/D$ and one of class $1$ has probability $(1-\lambda)/D$, so the
level sets of Sec.~\ref{sec:sectors} are the two classes. The client will attach to
each label its probability in units of $1/D$, the \emph{weight}
\begin{equation}
w(x)=D\,\pin(x)=1+\lambda(-1)^{c(x)} ,
\label{eq:weight}
\end{equation}
equal to $1+\lambda$ on class $0$ and $1-\lambda$ on class $1$
(Sec.~\ref{sec:protocol}).

\subsection{A spectrum whose sectors are cosets of a hidden parity}\label{sec:parity}

On a support $b$ we write $s|_b$ for the bits of $s$ on $b$, and $\ker s|_b$ for the
support labels $\vec x$ with $\langle s|_b,\vec x\rangle=0$. The sectors of the hidden
state then take a simple form.

\begin{lemma}[Hidden-parity spectrum]\label{lem:parity}
Let $\sigma_0$ be as in Eq.~\eqref{eq:rho0} with $\lambda\neq0$. (a) On a support $b$
with $s|_b\neq0$, two support labels are degenerate if and only if
$\langle s|_b,\vec x\rangle=\langle s|_b,\vec x'\rangle$: the sectors are the two
cosets of $\ker s|_b$, each of $2^{\ell-1}$ labels. (b) For every nonzero
$\omega\in\mathbb F_2^n$, $\mathbb E_{\pin}[(-1)^{\langle \omega,x\rangle}]=\lambda\,\delta_{\omega,s}$;
in particular every marginal on a set of qubits not containing the support of $s$ is
exactly uniform.
\end{lemma}

The proof is in Appendix~\ref{app:factor}. Part (a) replaces the exponential
comparison of conditional states that the definition of a sector suggests by one
inner product over $\mathbb F_2$, exact. Part (b) says that the hidden state is not a
product: its whole structure is a single many-body $Z$ string of weight $|s|\ge3$, and
every lower-order $Z$ coefficient vanishes. The sectors depend on $\lambda$ only
through $\lambda\neq0$; so do the circuit and the client's cost.

\subsection{Layers}\label{sec:layers}

The client builds each layer from a short record. After $t$ layers, every label $y$ of
the current frame descends from input labels of a single class
(Corollary~\ref{cor:confine}), and that class is again a parity, up to a constant:
\begin{equation}
c_t(y)=\langle s^{(t)},y\rangle\oplus\beta_t .
\label{eq:classt}
\end{equation}
The string $s^{(t)}$, the \emph{form}, and the bit $\beta_t$, the \emph{flip}, start at
$s^{(0)}=s$ and $\beta_0=0$, and they are the client's entire record of the hidden
state; the update below keeps Eq.~\eqref{eq:classt} true at every layer. We write
$\sL=s^{(L)}$ and $\beta=\beta_L$ for their final values and $c_L(y)$ for the class of
an outcome.

The support $b_t$ of $\ell=3$ qubits is drawn so as to straddle the central cut,
subject to $s^{(t-1)}|_{b_t}\neq0$. Its two sectors follow from
Lemma~\ref{lem:parity}(a), four labels each. Then:

\emph{The block factor.} The gate acts on the three qubits of the support, which carry
eight labels, and each sector holds four of them. Listing the four labels of a sector
in increasing order, we read the sector as the basis $\ket{00},\ket{01},\ket{10},
\ket{11}$ of a two-qubit space. These two \emph{sector qubits} are a relabeling, not two
of the physical qubits: the block, like the whole gate, acts on all three qubits of the
support. On each sector the block applies a two-qubit entangler in these coordinates.
Let $|\Phi^+\rangle=(|00\rangle+|11\rangle)/\sqrt{2}$ and define
$$
W_B|jk\rangle=(Z^j\otimes X^k)|\Phi^+\rangle,
\qquad
W=(I\otimes H)W_B .
$$
The columns of $W$ are maximally entangled states of the two sector qubits, and its
sixteen entries all have modulus $1/2$. The block applies $W$ on one sector and
$W\cdot\mathrm{CZ}$ on the other. This is where four-label sectors, and hence
$\ell\ge3$, are needed: with $\ell=2$ a sector holds only two labels, too few for a
two-qubit entangler. Each map is Clifford on its own sector, but the block that chooses between
them according to the class is not: in the sector coordinates it is $(I\otimes W)$
followed by a CZ controlled by the class bit, a CCZ gate up to Clifford gates, and it
maps only three of the six single-qubit Pauli generators on the support to Pauli
operators. The two maps are not equal, but on computational-basis inputs their
corresponding columns differ only by phase,
$(W\,\mathrm{CZ})|jk\rangle=(-1)^{jk}W|jk\rangle$.
Thus they have identical rank-one projectors, a property that matters for
the LBF of Ref.~\cite{Cable_2015} (Sec.~\ref{sec:concsim}). Among its candidate
projectors the LBF tries the
rank-one ones, $\ket{\vec x}\!\bra{\vec x}$, and such a candidate yields a valid reading
exactly when $G_tU_{t-1}\ket{\vec x}$ is a product state. If one branch of the block were trivial,
the columns belonging to that sector would be product states and the others would
not, and the pattern of valid rank-one candidates would display the sector
partition. With identical corresponding column projectors on both sectors and affine
permutations, all $d_b$ columns share the same local-unitary invariants and the
pattern carries no information. Identical entanglers would also achieve this, but
leave the Schmidt rank capped below its maximum; the pair $(W,W\cdot\mathrm{CZ})$
removes the cap. Taking $W$ dense, rather than a bare Bell map with two nonzero entries
per column, makes the LBF halt on every layer we scanned instead of half of them; the
price, quantified in Sec.~\ref{sec:instances}, is that truncated simulations keep more
of the input--output correlation. Table~\ref{tab:design} gives the measurements.

\emph{The permutation} $P_t$ is an invertible affine map $\vec x\mapsto K_t\vec x\oplus u_t$
on the support, with an invertible $\ell\times\ell$ matrix $K_t$ over $\mathbb F_2$ and
an offset $u_t$. It carries each coset of $\ker s^{(t-1)}|_{b_t}$ onto a coset of
$\ker s^{(t)}|_{b_t}$, so the class remains a parity, with
\begin{equation}
s^{(t)}|_{b_t}=\bigl(K_t^{-1}\bigr)^{\mathsf T}s^{(t-1)}|_{b_t},\qquad
\beta_t=\beta_{t-1}\oplus\langle s^{(t)}|_{b_t},u_t\rangle .
\label{eq:update}
\end{equation}
The affine map is drawn uniformly, independently of the secret. About one layer in
seven then leaves the restricted form unchanged, mapping the sector partition onto
itself; such a layer moves no level set, but its block still entangles. What must be
avoided is a circuit made only of such layers, which would leave every level set
invariant, so that a server could answer by returning the label it was asked to
prepare. In none of our instances is the final form equal to the secret.
Affine maps are used for two reasons: Eq.~\eqref{eq:update} is the client's whole
record-keeping, and an affine map treats the pairings inside both sectors alike,
which is what keeps the rank-one pattern uninformative (Table~\ref{tab:design}).

\emph{The local frames} $U_t$ are Haar-random single-qubit unitaries on the support,
and so generic non-Clifford rotations. Being local, they cancel from the spectrum,
leave every Schmidt spectrum invariant, and commute with every bond-dimension
truncation, while making each compiled gate a dense $d_b\times d_b$ unitary.

The compiled gate is $G_t=U_tP_tB_tU_{t-1}^\dagger$; the first is $G_1=U_1P_1B_1$,
since $U_0=I$. The client records
$s^{(t)},\beta_t$ and the frame on the support; it never holds an object of size
$2^n$. Appendix~\ref{app:construction} gives the recipe in full, and
Fig.~\ref{fig:circuit} shows one instance end to end.

\begin{figure*}[t]
\includegraphics[width=\textwidth]{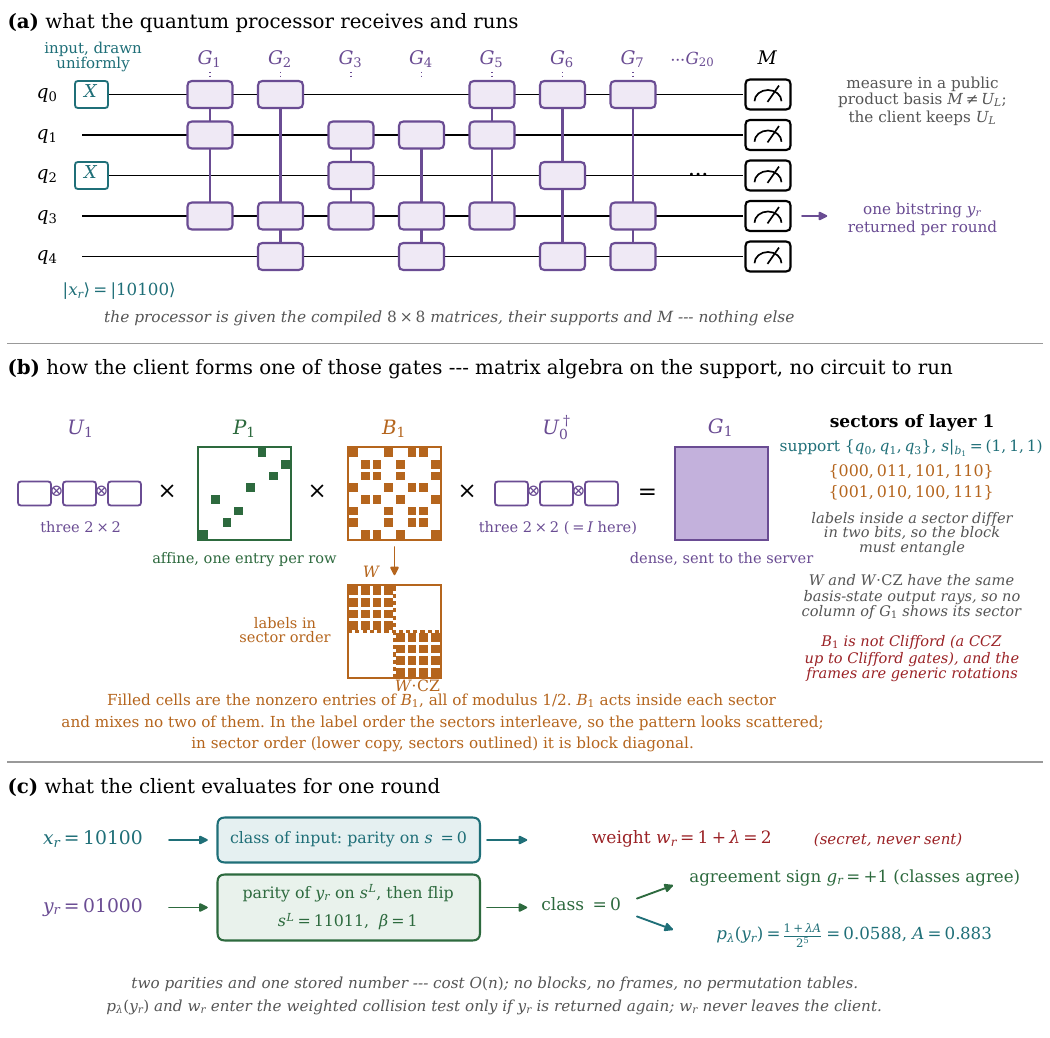}
\caption{One five-qubit instance of the construction, end to end ($\ell=3$, a secret
parity over all five qubits, $L=20$, $\lambda=1$, $\theta=0.15$). Every support,
sector, permutation and number shown is taken from that instance.
\emph{(a)} What the quantum processor receives and runs. The client sends a uniformly
random label, here $\ket{10100}$; the processor prepares it in the computational
basis, applies the
compiled gates $G_t$, each on three qubits straddling the central cut, measures in
the public basis $M\neq U_L$, and returns one bitstring. It is given the compiled
$8\times8$ matrices, their supports and $M$, and nothing else.
\emph{(b)} How the client forms $G_1=U_1P_1B_1U_0^\dagger$ on the support
$\{q_0,q_1,q_3\}$: a product of four $8\times8$ matrices, formed once; here $U_0=I$,
because the input is prepared in the computational basis. $P_1$ is the
affine permutation $(7,4,6,5,3,0,2,1)$ of the support labels; $B_1$ is $W$ on the
sector $\{000,011,101,110\}$ and $W\cdot\mathrm{CZ}$ on $\{001,010,100,111\}$, the
cosets of $s|_{b_1}=(1,1,1)$. Each sector's four labels are read as two sector
qubits, so $W$ and $W\cdot\mathrm{CZ}$ are $4\times4$ two-qubit maps, while $B_1$ and
$G_1$ act on all three physical qubits of the support. Filled cells are nonzero
entries, all of modulus $1/2$; the lower copy lists the labels in sector order, with
the two sectors outlined. $B_1$ is not Clifford, and the frames are generic rotations. Only the dense
product $G_1$ is sent.
\emph{(c)} What the client evaluates for one round. The class of the input fixes the
secret weight $w_r=2$. One parity of the outcome against the final form
$s^{L}=11011$, with flip $\beta=1$, fixes its class, hence the agreement sign
$g_r=+1$ and the target $\plam(y_r)=(1+\lambda A)/2^5=0.0588$, with $A=0.883$. The
outcome shown was sampled from the honest device. Neither the block factors, the
frames nor the permutation tables enter: the blocks cancel from the ensemble
(Theorem~\ref{thm:invariance}) and the permutations are summarized in $s^{L}$.}
\label{fig:circuit}
\end{figure*}

\begin{table}[t]
\caption{Choice of the block factor, at $\ell=3$. ``Leak'': layers, among all $208$
layers of four instances each at $n=6$ and $7$, on which the pattern of valid rank-one
LBF candidates equals the sector partition. ``Schmidt rank'': maximum over all inputs
across the central cut, at $n=8$ for four instances (ceiling 16). $W_B$ is a bare Bell
map and $W$ its dense version. On the first eight layers of three instances each at
$n=6$ and $8$, the LBF run as published halts on $31$, $32$, $24$ and $24$ of the $48$
layers for the four $W_B$ designs, in the order listed, and on all $48$ with $W$.}
\label{tab:design}
\begin{ruledtabular}
\begin{tabular}{lcc}
blocks, permutations & leak & Schmidt rank\\
\hline
$W_B$ and $I$, affine & $120/208$ & 16, 16, 16, 16\\
$W_B$ and $W_B$, non-affine & $149/208$ & 8, 16, 16, 16\\
$W_B$ and $W_B$, affine & $0/208$ & 16, 4, 8, 8\\
$W_B$ and $W_B\mathrm{CZ}$, affine & $0/208$ & 16, 16, 16, 16\\
$W$ and $W\mathrm{CZ}$, affine (used) & $0/208$ & 16, 16, 16, 16\\
\end{tabular}
\end{ruledtabular}
\end{table}

\begin{theorem}[Polynomial client]\label{thm:poly}
For a circuit built by the recipe of Appendix~\ref{app:construction} with $L$ layers
on supports of $\ell$ qubits:
\begin{enumerate}[label=(\roman*),leftmargin=*,itemsep=1pt]
\item the client stores $O(n+L\,d_b)$ numbers and never forms an object of size
$2^n$;
\item it builds the circuit in time $O(L\,[\,n\ell+d_b^{\,3}\,])$;
\item the sectors it uses are the exact ones at every layer, so the circuit satisfies
Theorem~\ref{thm:invariance} exactly;
\item it evaluates the weight, the agreement sign and the target probability of any
round in time $O(n)$;
\item attaching the local frames leaves the spectrum, the Schmidt coefficients of
every component across every bipartition, and every bond-dimension truncation
unchanged.
\end{enumerate}
\end{theorem}

The proof is in Appendix~\ref{app:poly}. Point (ii) is not at odds with the
observation that determining the degeneracy becomes hard once gates act on more than
two qubits~\cite{Eastin2010}: that difficulty belongs to a simulator, which must
infer the degeneracy from the gates. The client chose it.

\begin{remark}[Rejecting readable layers]\label{rem:reject}
The client can run the LBF itself while building, and redraw any layer the LBF reads
without ambiguity. With the dense block factors every layer we scanned was already
ambiguous, so the instances reported here did not need this option; it guarantees
the property for any other choice of blocks.
\end{remark}

% ==================================================================
\section{The protocol}\label{sec:protocol}

The construction above gives the client a compact private description of
the concordant ensemble while the server receives only the compiled gates.
We now turn that asymmetry into an explicit verification procedure.

The protocol runs $R$ rounds, indexed by $r$. In round $r$ the client sends an input
label $x_r$, the server returns an outcome $y_r$, and the client privately records a
weight $w_r$ for the round. For the honest device, $p(y|x)$ is the probability of
outcome $y$ given input $x$. Figure~\ref{fig:protocol} summarizes the exchange.

\begin{figure}[t]
\includegraphics[width=\columnwidth]{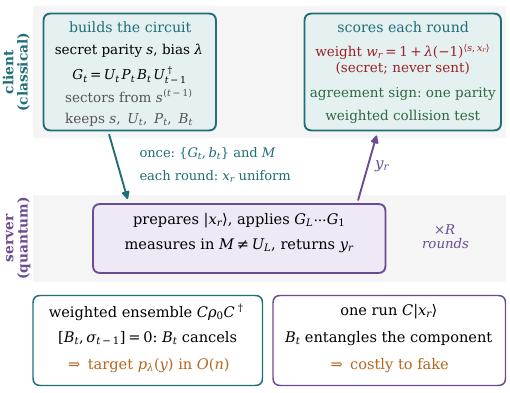}
\caption{The protocol. The client builds the circuit from a secret parity $s$ and
sends the compiled gates once. Each round it sends a uniformly random label $x_r$; the
server prepares $\ket{x_r}$, applies the gates and returns the outcome $y_r$ of a
measurement in $M\neq U_L$. The client keeps, for each round, the weight
$w_r=1+\lambda(-1)^{\langle s,x_r\rangle}$, the eigenvalue of $x_r$ in the hidden
state, and never sends it. Bottom: the weighted ensemble is $C\rho_0C^\dagger$, in
which the block factors cancel, so the target distribution costs $O(n)$ per outcome;
each run is a single component, which the blocks entangle.}
\label{fig:protocol}
\end{figure}

\subsection{Uniform inputs and secret eigenvalue weights}

A classical client must tell the server which basis state to prepare in each round,
so the label $x_r$ is transmitted in the clear. If the client drew $x_r$ from
$\pin$, the labels would be samples of the hidden state, and the class structure
that fixes every sector could be learned from them. The client draws $x_r$ uniformly
instead and applies the eigenvalues itself.

\begin{proposition}[Secret weighting]\label{prop:weights}
Let the server answer input $x$ with outcome $y$ by any strategy, possibly
randomized. For every function $f$ of an input and an outcome,
\begin{equation}
\mathbb E_{x\sim{\rm uniform}}\bigl[w(x)\,f(x,y)\bigr]
=\mathbb E_{x\sim\pin}\bigl[f(x,y)\bigr],
\label{eq:weighting}
\end{equation}
while the labels the server receives are independent of $s$ and $\lambda$.
\end{proposition}

The identity is $\pin(x)=2^{-n}w(x)$, Eq.~\eqref{eq:weight}. Physically it states how
mixed states work: the statistics of $\rho_0$ are the eigenvalue-weighted average of
the statistics of its eigencomponents. Each round runs one eigencomponent and the
client supplies its eigenvalue afterwards. The degenerate state, its sectors and the
blocks built on them are all in force; the eigenvalues that define them exist only in
the client's records. At $\lambda=1$ the weights are $0$ and $2$, and the client simply
discards, in secret, the rounds whose input lies in class $1$. The bias $\lambda$ is
therefore a free parameter of the tests rather than a quantity that must be kept small.

\subsection{The target distribution}\label{sec:target}

The device does not measure in the final frame $U_L$, which the client keeps, but in a
public product basis $M=\bigotimes_qM_q\neq U_L$. On qubit $q$ the two differ by the
\emph{tilt} $\Omega_q=U_{L,q}^\dagger M_q$, a single-qubit rotation. Its effect on a
measurement of $Z$ in the final frame is captured by one number,
$a_q=\bra0\Omega_q^\dagger Z\Omega_q\ket0$: the tilted measurement reports the same $Z$
value as an untilted one with probability $(1+a_q)/2$. The class of an outcome is a
parity of the qubits selected by the final form, so the part of it that survives the
tilt is the \emph{visibility}
\begin{equation}
A=\prod_{q:\,\sL_q=1}a_q .
\label{eq:visibility}
\end{equation}
Measuring the weighted ensemble $C\rho_0C^\dagger$ in the basis $M$ gives, by
Theorem~\ref{thm:invariance} and because each $\Omega_q^\dagger Z\Omega_q$ is traceless with
diagonal $(a_q,-a_q)$, the outcome distribution of the hidden state,
\begin{equation}
\plam(y)=\frac1D\Bigl[1+\lambda A\,(-1)^{c_L(y)}\Bigr].
\label{eq:qscore}
\end{equation}
This is one parity and a precomputed number $A$: time $O(n)$ per outcome. Weighting
the outcomes of an honest device by $w(x_r)$ reproduces Eq.~\eqref{eq:qscore} to
$10^{-16}$ in our instances. We take $M_q=U_{L,q}\,e^{-i\theta\,\hat n_q\cdot
\vec\sigma}$, a rotation by the tilt angle $\theta$ about a random unit axis $\hat n_q$,
with $\vec\sigma$ the vector of Pauli matrices. Then $\Omega_q=e^{-i\theta\,\hat n_q\cdot
\vec\sigma}$ turns the Bloch vector of $Z$ by $2\theta$ about $\hat n_q$, and
$a_q=1-2\sin^2\theta\,(1-n_{q,z}^2)$ lies between $\cos2\theta$ and $1$. $A$ is a
product over the $|\sL|\approx n/2$ qubits of the final form, and $\theta\sim1/\sqrt n$
keeps it of order one. Because $M$
is close to $U_L$ the server knows the final frame approximately. This is the mirror
image of the known input basis and, like it, supplies no sectors.

\subsection{The agreement test}\label{sec:agreetest}

The client knows the class of every input it sent. With
\begin{equation}
g_r=(-1)^{\,c(x_r)+c_L(y_r)},\qquad
\hat S=\frac1R\sum_{r=1}^Rg_r ,
\label{eq:agree}
\end{equation}
it accepts if
\begin{equation}
\bigl|\hat S-A\bigr|\le\tau\sqrt{(1-A^2)/R},
\label{eq:agreerule}
\end{equation}
where $\tau$ is a rejection threshold, counted in standard deviations of the honest
statistic.

\begin{proposition}[Agreement]\label{prop:agree}
For the honest device the $g_r$ are independent $\pm1$ variables of mean $A$, for every
$\lambda$. If the outcome is independent of the input, the mean is zero.
\end{proposition}

The proof is in Appendix~\ref{app:agree}. The circuit carries the class operator of the
input onto that of the output, $C\,Z^s\,C^\dagger=(-1)^\beta U_L
Z^{\sL}U_L^\dagger$, and the rotation to $M$ dilutes the correlation to $A$. The test is
two-sided on purpose: a server that returns the correct class in every round has
$\hat S=1$ and is rejected. Reproducing $A$ requires the class map and the final frame
together. A spoofer whose expected agreement is $\bar S$ is rejected after about $\tau^2(1-\bar S^2)/(A-\bar S)^2$ rounds;
a server whose answers ignore the input is rejected after $\tau^2/A^2$.

\subsection{The weighted collision test}\label{sec:collisiontest}

The agreement test reads one Fourier coefficient of the joint distribution of input
and outcome. The second test reads the shape of the outcome distribution. For each
string $y$ returned in the run, let $S_1(y)$ and $S_2(y)$ be the sums of $w_r$ and of
$w_r^2$ over the rounds that returned it, and $\mathcal W=\sum_rw_r$ the total weight. The statistic is
\begin{equation}
Z_w=\sum_y\frac{S_1(y)^2-S_2(y)}{(R-1)\,\plam(y)}-R-2\,(\mathcal W-R).
\label{eq:Zw}
\end{equation}
The numerator $S_1^2-S_2$ is twice the sum, over pairs of different rounds that
returned $y$, of the product of their two weights, so only strings returned at least
twice contribute; each needs $\plam(y)$ once, in time $O(n)$. The last term removes the
fluctuation of the total weight, which has zero mean for every server because the
inputs are uniform.

\begin{proposition}[Weighted collision statistic]\label{prop:collision}
(a) For any server, $\mathbb E[Z_w]=R\,\chi^2(\tilde p\,\|\,\plam)$, where
$\chi^2(p\|q)=\sum_y[p(y)-q(y)]^2/q(y)$ is the chi-square divergence and
$\tilde p(y)=\mathbb E_{x\sim{\rm uniform}}[w(x)\,p_{\rm server}(y|x)]$ is the outcome
distribution the server's strategy would produce on the hidden state. (b) For the
honest device the variance $\sigma_Z^2$ of $Z_w$ is
\begin{equation}
\sigma_Z^2=\frac{2R}{R-1}\Bigl[\frac D2\bigl(F_++F_-\bigr)-1-2\lambda^2\Bigr],
\label{eq:sigma}
\end{equation}
where $F_\pm=(1+\lambda^2\pm2\lambda A)^2/(1\pm\lambda A)^2$.
At $\lambda=0$ this is the unweighted value $2(D-1)R/(R-1)$; at $\lambda=1$ it is
$2(4D-3)R/(R-1)$.
\end{proposition}

The proof, in Appendix~\ref{app:collision}, writes $Z_w$ as a degenerate $U$-statistic of
order two. The client rejects if $Z_w>\tau\sigma_Z$. On the instances of
Sec.~\ref{sec:instances}, $400$ simulated honest runs at $\lambda=1$ give a mean of
$-2.2\pm2.2$ and a spread of $43.1$, against $0$ and $45.2$ from
Eq.~\eqref{eq:sigma}, consistent within sampling error. A spoofer is rejected after about
\begin{equation}
R_{\rm agr}\simeq\frac{\tau^2(1-\bar S^2)}{(A-\bar S)^2},\qquad
R_{\rm col}\simeq\frac{\tau\,\sigma_Z}{\chi^2(\tilde p\,\|\,\plam)}
\label{eq:rounds}
\end{equation}
rounds of the respective test. Since $\sigma_Z\simeq\sqrt{8D}$ at $\lambda=1$, $R_{\rm col}$
grows as $\sqrt D$: at $n=50$ a spoofer whose distribution deviates by
$\chi^2(\tilde p\|\plam)=0.3$ reaches the threshold after about $1.6\times10^9$ rounds, and is rejected
reliably after a few times that. This is the main cost of
the protocol. The client's share of it is $O(Rn)$ elementary operations and a table of
the returned strings.

\subsection{What each party holds}

The client accepts if both tests pass. The server receives the \emph{public data}:
the compiled gates with their supports, the basis $M$, and one
uniformly random label per round. The client keeps everything else: the secret $s$
and the bias $\lambda$, the factors $U_t$, $P_t$ and $B_t$, the final frame $U_L$, the
final form $\sL$ with its flip $\beta$, and the weight of every round.

% ==================================================================
\section{Soundness}\label{sec:soundness}
% ==================================================================

\subsection{The secret must stay hidden}

The first question is what a server could do with the secret. The answer fixes what
soundness must mean.

\begin{proposition}[Known sectors make the reading unique]\label{prop:known}
Suppose the sector projectors $\{\Pi_k\}$ of $\sigma_{t-1}$ on $b_t$ and the frame
$U_{t-1}$ on $b_t$ are known. Any product basis in which every
$G_tU_{t-1}\Pi_kU_{t-1}^\dagger G_t^\dagger$ is diagonal reproduces $\sigma_t$. The
remaining freedom, a relabeling inside a sector or a basis change on a qubit outside
the current form, is a symmetry of the state and leaves $\sigma_t$
unchanged~\cite{Cable_2015}.
\end{proposition}

Finding that basis is linear algebra on $d_b\times d_b$ matrices: for each qubit one
solves for the $2\times2$ projectors commuting with every
$G_tU_{t-1}\Pi_kU_{t-1}^\dagger G_t^\dagger$. The image of each sector under the
recovered map gives the permutation, and Eq.~\eqref{eq:update} the next form.

\begin{proposition}[Holding the secret suffices]\label{prop:trapdoor}
A classical algorithm given $s$ in addition to the public data runs in time
$O(L\,2^{3\ell}+Rn)$ and passes both tests: its agreement statistic has exactly the honest
distribution, and its weighted collision statistic has the honest mean and variance.
\end{proposition}

The proof is in Appendix~\ref{app:trapdoor}. Starting from the computational basis and the sectors given
by $s$, Proposition~\ref{prop:known} recovers every layer, hence $\sL$, $\beta$ and
$U_L$, and so $A$. The algorithm then answers $x$ with a uniformly random $y$ whose
class agrees with that of $x$ with probability $(1+A)/2$. Supplied with the true
sectors, the gate-by-gate reading reproduces the output distribution to total
variation $1\times10^{-14}$ (Sec.~\ref{sec:instances}). Soundness is therefore exactly
the statement that the server cannot obtain $s$, or any equivalent description of the
class map together with $A$, from what it is given.

\subsection{Soundness in an ideal model}\label{sec:ideal}

Suppose first that the server's view carried no information about the hidden data.
In the affine construction the class of an input and the class of the outcome it
leads to are tied by a linear relation fixed by the composite permutation. Let
$\kappa$ be the largest bias, under the client's sampling, of any nonzero linear
function of the pair $(s,\sL)$; it is of order $2^{-n}$ when $s$ is uniform and the
composite permutation mixes well (Appendix~\ref{app:ideal}).

\begin{theorem}[Soundness in the ideal model]\label{thm:ideal}
If the server's view is independent of the hidden data ($s$, the permutations and
$U_L$), then for every strategy
\begin{equation*}
\mathbb E\bigl[\hat S^2\bigr]\le\frac1R+2^{-n}+\kappa ,
\end{equation*}
and the agreement test accepts with probability at most
$(1/R+2^{-n}+\kappa)/(A-\delta_R)^2$ whenever $A>\delta_R=\tau\sqrt{(1-A^2)/R}$.
\end{theorem}

The proof is in Appendix~\ref{app:ideal}. In the ideal model the agreement test alone
is sound. The collision test is needed for the explicit circuit, because a simulator
that reads the gates can reproduce the agreement correlation without reproducing the
distribution (Sec.~\ref{sec:mps}).

\subsection{The explicit circuit against known simulation methods}

The compiled gates are not independent of the hidden data: they are built from it.
Soundness of the explicit construction is therefore a statement about whether the
hidden data can be extracted from them, the question that decided the fate of the
IQP schemes~\cite{KahanamokuMeyer2019,GrossHangleiter2025}. We take the known routes to simulating
concordant computation in turn.

\paragraph*{Label-tracking simulation.}
As reviewed in Sec.~\ref{sec:concsim}, the efficient algorithms of
Ref.~\cite{Eastin2010} follow labels in settings where the relevant degeneracy
structure can be identified efficiently. Our circuits lie outside those settings
(Table~\ref{tab:known}): they use three-qubit gates on a degenerate, non-product,
full-rank input, and the server is not given that input's degeneracy structure.

\paragraph*{The single-qubit symmetry test.}
One ingredient of the analysis of Ref.~\cite{Cable_2015} is the single-qubit symmetry
criterion: for a factorized input and a Clifford permutation $S$, $S$ is a symmetry
of $\sigma$ if and only if $\mathrm{Tr}[(S\sigma S^\dagger-\sigma)Z_j]=0$ for every
qubit $j$ (Theorem~III.4 of Ref.~\cite{Cable_2015}). The factorization is essential.

\begin{theorem}[The single-qubit test needs a product input]\label{thm:symtest}
Let $\sigma$ be diagonal on $n\ge2$ qubits with $\mathrm{Tr}[\sigma Z_j]=0$ for every
$j$ and $\sigma\neq I/2^n$. Then
\begin{enumerate}[label=(\alph*),leftmargin=*,itemsep=1pt]
\item $\sigma$ is not a product state;
\item the test statistic vanishes for every bit-flip string $X_T$, $T\subseteq[n]$,
each of which is a Clifford permutation, yet at least one $X_T$ is not a symmetry.
\end{enumerate}
For the hidden-parity state $\sigma_0$ with $|s|\ge2$:
\begin{enumerate}[label=(\alph*),leftmargin=*,itemsep=1pt,start=3]
\item $X_T$ is a symmetry if and only if $\langle s,T\rangle=0$, so half of the $2^n$
bit-flip strings are false positives;
\item for an affine permutation $x\mapsto Kx\oplus u$, write $s'=(K^{-1})^{\mathsf T}s$. The
statistic vanishes if and only if $|s'|\neq1$, and $S$ is a symmetry if and only if
$s'=s$ and $\langle s,u\rangle=0$. A fraction $1/[2(2^n-1-n)]$ of the affine
permutations that pass the test are symmetries.
\end{enumerate}
\end{theorem}

The proof is in Appendix~\ref{app:symtest}. We checked (d) by exhaustive enumeration at
$n=4$: $236\,544$ of the $322\,560$ affine permutations pass the test and $10\,752$, a
fraction $1/22$, are symmetries. Part (d) covers exactly the affine permutations of our
construction: the test fails because of the input, whatever the permutations.
The efficient degeneracy test of Ref.~\cite{Eastin2010} relies on a product input and
on access to the input distribution, whereas our input is correlated and its
distribution is hidden from the server. Theorem~\ref{thm:symtest} adds a further
obstruction: the single-qubit symmetry criterion of Ref.~\cite{Cable_2015} gives false
positives on this input even when its description is supplied.

\paragraph*{The local-basis finder.}
We next apply the LBF reviewed in Sec.~\ref{sec:concsim} directly to the
compiled gates of our construction~\cite{Cable_2015}. By
Proposition~\ref{prop:known}, once the true sectors are supplied the basis
can be recovered; the question here is what happens when they are not.
Run as published on our circuits, the LBF halts on every
one of the $48$ layers we scanned; with bare Bell-map blocks it halted on half of
them. With Haar-random blocks it halts on none, and on layers that are pure
permutations it halts on none either: there all $162$ candidates are valid and
mutually compatible. The rank-one candidates fall into a single class of local-unitary
invariants on every layer, so they carry no information about the partition. The
halting is a property of each layer. It is necessary for soundness but not
sufficient, since it does not by itself exclude methods that combine the readings of
overlapping layers; no such method appears in the relevant
literature~\cite{Eastin2010,Cable_2015}.

\paragraph*{Matrix-product simulation.}\label{sec:mps}
Its cost is governed by the Schmidt rank across the cuts of the
register~\cite{Vidal2003,Jozsa2003}. We write $\chi$ for the bond dimension, the number
of Schmidt values a truncated simulation keeps at each cut; it is unrelated to the
divergence $\chi^2$ of Sec.~\ref{sec:collisiontest}. The transmitted component saturates the maximal
rank $2^{n/2}$ on every $n=8$ instance (Sec.~\ref{sec:instances}). A truncated
simulation is not unitary and loses part of the input--output correlation.
Table~\ref{tab:sweep} shows both effects. Small bond dimensions are rejected quickly by
both tests. Already at a quarter of the maximal bond dimension ($\chi=8$ at $n=10$) the
truncated simulation retains the whole agreement correlation, so the agreement test is
blind, while the weighted collision test still rejects it, after about $1.3\times10^4$
rounds.
That is why the protocol uses both.

\paragraph*{Stabilizer simulation.}
Clifford circuits are simulated efficiently~\cite{Gottesman1998,Aaronson2004}, and
circuits with few non-Clifford gates at a cost exponential only in their
number~\cite{Bravyi2016}. Our circuits are not Clifford, and not at either level. The
permutations are affine and hence Clifford, but each block factor is a CCZ gate up to
Clifford gates (Sec.~\ref{sec:layers}), and every layer attaches three Haar-random
single-qubit frames. Even with the frames stripped off, a circuit of depth
$L=\Theta(n)$ therefore contains $\Theta(n)$ non-Clifford gates, and stabilizer-rank
methods scale exponentially in that number; stripping the frames would in any case
require recovering them, which is the LBF problem.

\subsection{Relation to earlier simulation results}

The simulations of Refs.~\cite{Eastin2010,Cable_2015} are given the input state. In that
setting our circuits are easy: a simulator holding $\rho_0$ knows $s$, and
Proposition~\ref{prop:trapdoor} then simulates the circuit gate by gate. The reason is
general. The client's record-keeping, Eq.~\eqref{eq:update}, consumes only the current
form and the permutation of the layer, and the permutation can be read from the gate
once the sectors are known; whoever holds the starting point can replay the client's
construction. The hardness here is therefore of a different kind from the one those
works leave open: it comes from withholding the input's degeneracy structure, not from
the dynamics alone. Our circuits are a concrete case in which the published LBF halts,
and they show that such halting is not by itself evidence of hardness when the input
is known and structured.

\subsection{Conjectures}

\begin{conjecture}[Concordant computation is not efficiently simulable in general]
\label{conj:general}
There exist uniformly generated families of concordance-preserving circuits on $n$
qubits, with inputs of polynomial description given to the simulator, for which no
classical algorithm running in time polynomial in $n$ samples the output of the
individual runs to within total variation $1/\mathrm{poly}(n)$.
\end{conjecture}

The conjecture cannot hold for all concordant computations, since the cases of
Refs.~\cite{Eastin2010,Cable_2015} are efficiently simulable; it takes the open side of
the question those works leave. By the previous subsection, the family of this paper
does not bear on it. The statement the protocol rests on is the following.

\begin{definition}[Hidden secret]\label{def:hidden}
A family of circuits of the form of Sec.~\ref{sec:construction} \emph{hides its
secret} if no classical algorithm running in time polynomial in $n$, given the public
data, outputs $s$, or an equivalent description of the class map, together with $A$,
except with probability negligible in $n$.
\end{definition}

\begin{conjecture}[Soundness of constructions that hide the secret]\label{conj:main}
For a family that hides its secret, used with uniform inputs, secret weights and
$M\neq U_L$, no classical algorithm running in time polynomial in $n$, given the
public data, passes both tests at the round counts of Eq.~\eqref{eq:rounds} except
with probability $\varepsilon(n)$ negligible in $n$.
\end{conjecture}

By Proposition~\ref{prop:trapdoor}, hiding the secret is necessary for soundness; the
conjecture is that it is also sufficient, so that the compiled circuit of such a family
is as good as the ideal model of Theorem~\ref{thm:ideal} against polynomial-time
classical servers. The recipe of Appendix~\ref{app:construction} is one way to build a
circuit of this form; other choices of blocks and permutations are possible, and the
client may use any of them. For the recipe given here, no known classical algorithm in
the relevant literature~\cite{Eastin2010,Cable_2015} recovers the secret from the
public data; we do not claim that it hides its secret in the sense of
Definition~\ref{def:hidden}. Whether a family that does exists, with a client that
still tracks the class map in polynomial time, is open.

% ==================================================================
\section{Instances and measurements}\label{sec:instances}
% ==================================================================

All instances are generated by the recipe of Appendix~\ref{app:construction} with
$\ell=3$, $L=4n$, a secret of weight five, the block factors $W$ and $W\mathrm{CZ}$, and
affine permutations.

\emph{The construction.} On three instances each at $n=6$ and $n=8$, ensemble
invariance (Theorem~\ref{thm:invariance}) holds to $7\times10^{-16}$ and branch
confinement (Corollary~\ref{cor:confine}) is exactly one. The transmitted component
saturates the maximal Schmidt rank $2^{n/2}$ on all three instances at $n=8$ and on two
of the three at $n=6$. The LBF, run as published, halts on all of the first $48$
layers; with Haar-random blocks, or on layers that are pure permutations, it halts on
none. On every layer the rank-one candidates fall into a single class of local-unitary
invariants. Supplied with the true sectors, the gate-by-gate reading of
Proposition~\ref{prop:known} reproduces the output distribution of an $n=7$ instance to
total variation $1\times10^{-14}$.

\emph{The protocol.} Table~\ref{tab:sweep} runs it against spoofers.

\begin{table}[t]
\caption{Estimated rounds needed to reject a spoofer, at $\lambda=1$,
$\theta=0.15$ and $\tau=5$, from Eq.~\eqref{eq:rounds}.
Rejection in $95\%$ of simulated runs took up to four times as
many rounds. The $n=8$ results cover two independently generated
circuit instances with different secret parities (seeds 0 and 1,
$A=0.895$ and $0.864$); numerical ranges give the smaller and
larger estimates across these instances. The $n=10$ results
refer to one instance (seed 0, $A=0.766$). ``Uniform'' returns
random strings, ``echo'' returns the input label, and
``$\chi$'' denotes an MPS simulation truncated to bond dimension
$\chi$ after every gate, with ties between equal Schmidt values
broken by the numerical decomposition. ``Blind'' is the
numerical label used when the estimated agreement rejection
count exceeds $10^6$ rounds or is infinite; it does not by itself
establish exact equality of the honest and spoofer means.
Thus ``$4\,537$; blind'' reports approximately $4\,537$ agreement
rounds for one eight-qubit instance and a count exceeding this
reporting threshold for the other. The collision test gives
finite rejection estimates for both.}
\label{tab:sweep}
\begin{ruledtabular}
\begin{tabular}{lcccc}
 & \multicolumn{2}{c}{$n=8$} & \multicolumn{2}{c}{$n=10$}\\
spoofer & agreement & collision & agreement & collision\\
\hline
uniform & 31--34 & 56--77 & 43 & 319\\
echo & 31--34 & 25--33 & 43 & 118\\
$\chi=4$ & 62--344 & 189--1\,158 & 129 & 905\\
$\chi=8$ & 4\,537; blind & 9\,195--12\,025 & blind & 13\,124\\
\end{tabular}
\end{ruledtabular}
\end{table}

Three features of Table~\ref{tab:sweep} matter. The trivial spoofers fail the agreement
test within a few tens of rounds and, because the target is far from flat at
$\lambda=1$, the collision test as well. Truncated simulations lose the agreement
correlation at small bond dimension but keep it from a quarter of the maximal bond
dimension on, so against them the collision test carries the rejection, at about
$10^4$ rounds at these sizes. With bare Bell-map blocks the agreement test held out
longer and the collision test needed fewer rounds, $2\,338$ instead of $13\,124$ at
$\chi=8$ and $n=10$; we accept that cost in exchange for an LBF that halts on every
layer. The entanglement of the components is
the same for both choices, $3.2$ of $4$ bits across the central cut at $n=8$ and
$4.25$ of $5$ at $n=10$. How the bond dimension needed to pass scales with $n$ is not
settled by sizes up to ten.

% ==================================================================
\section{Randomness}\label{sec:entropy}
% ==================================================================

The honest outcome of a round is spread over many strings. Its Born-rule min-entropy
given the input, $H_{\min}(x)=-\log_2\max_yp(y|x)$ in the basis $M$, averages $4.2$ bits
at $n=8$ and $5.8$ bits at $n=10$, with worst inputs at $2.8$ and $4.5$ bits: about
$n-4$ bits per round. No test can certify more than the honest state provides.

Neither test by itself certifies that an outcome was unpredictable: by
Proposition~\ref{prop:trapdoor}, a server holding the secret answers from a classical
rule and passes. The entropy must come from the conjecture, because a classical
predictor of the device's outcomes would itself be a classical algorithm that passes.
It therefore applies to constructions that hide their secret
(Definition~\ref{def:hidden}).

\begin{proposition}[Entropy from soundness]\label{prop:soundentropy}
Suppose a classical polynomial-time server passes with probability at most
$\varepsilon(n)$, as in Conjecture~\ref{conj:main}. Let a device pass with probability
$\Pr[\mathrm{acc}]$, returning $Y=(y_1,\dots,y_R)$, and let an adversary hold side
information $E$ computed by a classical polynomial-time algorithm from the public data,
the round labels and randomness it may share with the device. Then, against any
classical polynomial-time predictor,
\begin{equation}
H_{\min}(Y\,|\,E,\mathrm{acc})\;\ge\;\log_2\frac1{\varepsilon(n)}
-\log_2\frac1{\Pr[\mathrm{acc}]} .
\label{eq:soundentropy}
\end{equation}
\end{proposition}

The proof is in Appendix~\ref{app:entropy}. The bound concerns the whole run and is set
by the soundness error; a per-round rate needs a further assumption, exactly as for
cross-entropy benchmarking~\cite{AaronsonHung2023,AaronsonGunn2020}.

\begin{conjecture}[Long-list hardness]\label{conj:longlist}
For instances of a construction that hides its secret, no classical algorithm running
within time $T_{\max}$, given the public data and an input $x$,
outputs a list of $J\le2^\eta$ labels with $\frac1J\sum_iD\,p(y_i|x)\ge1+\delta$, except
with negligible probability.
\end{conjecture}

Under Conjecture~\ref{conj:longlist}, each accepted round carries at least
$\min\{\eta,H_{\min}(x)\}-\log_2(1/\Pr[\mathrm{acc}])$ bits. Turning per-round statements
into a guarantee on the total entropy of $R$ rounds, against an adversary whose strategy
may depend on the whole transcript, is what entropy accumulation provides in the
device-independent setting~\cite{Dupuis2020,Metger2022}; we have not carried that
analysis out.

\paragraph*{Scope.} All statements above concern classical adversaries. A server that
operates the device sees its own input--outcome pairs, and each pair satisfies
$c(x_r)=c_L(y_r)$ with probability $(1+A)/2$. With $A$ of order one and many more than $n$ rounds, these noisy linear
equations in the $2n$ unknown bits of $(s,\sL)$ can be solved, after which such a server
could switch to the algorithm of Proposition~\ref{prop:trapdoor} with outcomes it
predicts. Certifying randomness against that adversary requires a fresh secret for
every round, which the collision test, as it stands, does not allow
(Sec.~\ref{sec:discussion}).

% ==================================================================
\section{Discussion}\label{sec:discussion}
% ==================================================================

The protocol rests on one asymmetry and one secret. The asymmetry is that block
factors commute with a degenerate concordant state, so they cancel from the ensemble
and act on every run. The secret is a parity: given it, the sectors of every layer
follow by one inner product, the client builds the circuit and predicts the ensemble
cheaply, and, by Proposition~\ref{prop:trapdoor}, so would anyone else. Everything in
the protocol serves to keep that one secret where it belongs. The inputs are uniform
and the eigenvalues are applied in the client's bookkeeping, so the labels carry
nothing; the block factors and permutations are chosen so that no single gate
displays the sectors; and the only channel left to a classical server is the compiled
circuit, which is the content of Conjecture~\ref{conj:main}.

That conjecture plays the role that the hardness of secret extraction plays for the
IQP schemes, and the history of those schemes is the reason to treat it with
care~\cite{ShepherdBremner2009,KahanamokuMeyer2019,YungCheng2020,BremnerChengJi2025,
GrossHangleiter2025}. There the secret is a vector hidden in the published program, a
binary matrix, and the attacks recover it from the algebraic structure of that
matrix. Here no binary matrix carrying the secret is published: the published object
is a sequence of dense $8\times8$ unitaries in random local frames, and what an
attacker must recover is which subspaces each gate preserves. 

% The LBF halts on every
% layer we scanned when it is not told, and no known classical algorithm in the relevant
% literature~\cite{Eastin2010,Cable_2015} recovers the secret of the explicit
% construction. The sharpest open question is whether there is a construction whose
% public data carry no signature of the secret while the client still tracks the class
% map in polynomial time. Any such proposal should be tested against methods that
% combine the readings of overlapping layers, and publishing challenge instances, as was
% done, and answered, for the first IQP
% scheme~\cite{ShepherdBremner2009,KahanamokuMeyer2019}, remains the natural way to test
% it.

The block factors and permutations of the explicit construction were chosen against
the simulation methods known in the literature~\cite{Eastin2010,Cable_2015}, which
read the circuit one gate at a time: the single-qubit symmetry test fails on the input
(Theorem~\ref{thm:symtest}), and the LBF halts on every layer we scanned when it is
not told the sectors (Table~\ref{tab:design}). No known classical algorithm in the
relevant literature recovers the secret of the explicit construction. Ambiguity on
each layer is, however, necessary rather than sufficient. A classical method could
carry every reading of each layer forward and discard those that later layers
contradict, so that the ambiguity narrows with depth; we have not found such a method
in the literature, and we do not claim that the explicit construction withstands one.
Because the circuit is entirely the client's design, a construction can be adapted to
such a method, as the block factors here were adapted to the LBF, for instance with
layers that restore ambiguity after consistency across earlier layers has narrowed it.
What such a construction must achieve is ambiguity that later layers cannot remove,
so that several readings with different class maps remain consistent with the whole
circuit, while the client still tracks the class map in polynomial time. No result we
are aware of rules this out, and whether it can be achieved is the sharpest open
question raised by this work. Publishing challenge instances, as was done, and
answered, for the first IQP scheme~\cite{ShepherdBremner2009,KahanamokuMeyer2019},
remains the natural way to test any such proposal.

The protocol moves the cost of verification. In certification by random circuit
sampling the client's classical work grows exponentially with the number of
qubits~\cite{AaronsonHung2023}; here it is $O(n)$ per round, so the classical side of
verification is no longer the bottleneck. What remains is the number of rounds. It is set
by the collision test, which, like any test of an output distribution that uses no
further structure, needs a number of samples growing as $\sqrt D$~\cite{Canonne2020}, and
by the fidelity of the device, since noise dilutes both statistics.

Reducing that cost is beyond the scope of this paper, but the construction points to
several directions. The agreement test checks a description of the output of polynomial
size and needs a number of rounds that does not grow with $n$. At the sizes studied here
it is reproduced by truncated simulations of large bond dimension, which is why the
collision test is used; whether it remains sound on its own against efficient
simulation as $n$ grows, so that the collision test could be dropped, is the most direct
question. Richer statistics of the same kind, from a hidden state with several secret
parities or from repeated inputs, would give a spoofer more to reproduce at little cost
in rounds. A test that needs few rounds would also allow a fresh secret in every round,
which closes the device-output channel of Sec.~\ref{sec:entropy}. Independently,
shallower circuits raise the fidelity and so lower the number of rounds; the depth used
here was chosen conservatively, and the smallest depth that keeps the properties of
Sec.~\ref{sec:soundness} is an open question about the construction.

The construction also admits an extension from a single secret parity to an
array of independent secret vectors $s_1,\ldots,s_k\in\mathbb F_2^n$.
Their joint parity values define $2^k$ classes, to which distinct weights
may be assigned to obtain more than two degenerate spectral sectors.
Affine layers update the secret array componentwise, while blocks acting
within the joint classes continue to cancel from the concordant ensemble.
Uniform inputs and private class weights therefore retain the same
verification principle, with the target distribution and agreement
statistics determined by the secret parities and their combinations.
For fixed $k$, or for $k=O(\log n)$ with polynomial-size gate matrices,
the sector records and verification tables remain polynomial in size;
gate supports must also retain nontrivial blocks within the sectors.
Increasing the number of sectors therefore increases the bookkeeping but
preserves polynomial classical cost for circuit construction and per-round
verification. The required number of rounds and the effect on spoofing
resistance are separate questions for future investigation.

Three further questions are open. Whether there are degeneracy structures other than
parities that a client can track locally without disclosing them is not known; the
parity is the simplest, and the attacks above use its linearity only through
Proposition~\ref{prop:trapdoor}. Whether concordant computation with a known input is
hard, Conjecture~\ref{conj:general}, is a separate question, on which this family does
not bear. And the effect of hardware noise, in particular coherent errors in the local
frames, on both statistics has not been analyzed; testing the instances on hardware is
the next step.

\begin{acknowledgments}
This work was supported by QpiAI India Private Limited. The first author sincerely thanks Naga Dileep Varikuthi of the University of Trento, Italy, for his valuable comments and suggestions on the manuscript. The first author gratefully acknowledges Vaibhav Madhok and Sourav Manna of the Indian Institute of Technology Madras for their valuable insights. The authors also acknowledge insightful discussions with Sachin Vaidya and Hari Hara Suthan Chittoor. 
\end{acknowledgments}

% ==================================================================
\section*{Disclosure of AI tool use}

Claude Opus 5 (Anthropic) was used for assistance with writing, organization, code
generation, and formalization of mathematical arguments. All scientific ideas, methods,
theorems, and proof strategies were developed by the author. All AI-assisted material
was independently verified by the author, who takes full responsibility for the content.

\bibliography{bibli2}
\bibliographystyle{apsrev4-2}

\onecolumngrid
\appendix

% ==================================================================
\section{Notation}\label{app:notation}
% ==================================================================

The symbols of the main text are listed in the order in which they appear, with the
section that introduces each.

\begin{center}
\renewcommand{\arraystretch}{1.2}
\begin{tabular}{p{0.25\textwidth}p{0.6\textwidth}l}
\hline\hline
symbol & meaning & Sec.\\
\hline
$n$; $D=2^n$ & number of qubits; number of computational-basis states &
\ref{sec:background}\\
$x$, $y$; $\oplus$ & bit-string labels of basis states; bitwise addition modulo two &
\ref{sec:background}\\
$\langle s,x\rangle$; $Z^s$ & parity of the bits of $x$ selected by $s$; product of
Pauli $Z$ on those qubits, with eigenvalue $(-1)^{\langle s,x\rangle}$ on $\ket x$ &
\ref{sec:background}\\
$L$, $t$; $G_t$ & number of layers, layer index; gate of layer $t$ &
\ref{sec:sectors}\\
$b_t$, $\ell$, $d_b=2^\ell$; $a$ & support of layer $t$, its number of qubits and of
labels; the remaining qubits & \ref{sec:sectors}\\
$\vec x$; $v$ & part of a label on the support; part on the rest & \ref{sec:sectors}\\
$\rho_0$, $\pin$; $\Lset(x)$ & input state and its distribution over labels; labels
with the same input probability as $x$ & \ref{sec:sectors}\\
$\rho_t$; $U_t$; $\sigma_t$ & state after $t$ layers; its frame, a product of
single-qubit unitaries; the state in its own frame, Eq.~\eqref{eq:sigmat} &
\ref{sec:sectors}\\
$\tilde\rho^{(a)}_{\vec x}$; $\mathcal S_i$, $\Pi^{(b)}_i$ & conditional state of the
rest given support label $\vec x$; sectors and their projectors,
Eq.~\eqref{eq:sectors} & \ref{sec:sectors}\\
$P_t$; $B_t$ & permutation of support labels; block factor, Eq.~\eqref{eq:gate} &
\ref{sec:cancellation}\\
$C$; $P$ & circuit $G_L\cdots G_1$; composite permutation $P_L\cdots P_1$ &
\ref{sec:cancellation}\\
$s$; $\lambda$ & secret parity, $|s|\ge3$; bias, $0<\lambda\le1$ &
\ref{sec:hiddenstate}\\
$c(x)$; $w(x)$ & class of an input label; its weight $1+\lambda(-1)^{c(x)}$,
Eq.~\eqref{eq:weight} & \ref{sec:hiddenstate}\\
$s|_b$; $\ker s|_b$ & bits of $s$ on a support; support labels of even parity under
them & \ref{sec:parity}\\
$s^{(t)}$, $\beta_t$; $c_t(y)$ & form and flip after $t$ layers; class of a label after
$t$ layers, Eq.~\eqref{eq:classt} & \ref{sec:layers}\\
$\sL$, $\beta$; $c_L(y)$ & final form and flip; class of an outcome &
\ref{sec:layers}\\
$W_B$; $W$ & Bell map; its dense version, applied on a sector & \ref{sec:layers}\\
$K_t$, $u_t$ & matrix and offset of the affine permutation $P_t$, Eq.~\eqref{eq:update}
& \ref{sec:layers}\\
$R$, $r$; $x_r$, $y_r$, $w_r$ & number of rounds, round index; input, outcome and weight
of round $r$ & \ref{sec:protocol}\\
$p(y|x)$ & outcome distribution of the honest device for input $x$ &
\ref{sec:protocol}\\
$M$; $\Omega_q$; $\theta$ & public measurement basis; tilt on qubit $q$; tilt angle &
\ref{sec:target}\\
$a_q$; $A$ & part of the $Z$ value of qubit $q$ that survives the tilt; visibility,
Eq.~\eqref{eq:visibility} & \ref{sec:target}\\
$\plam(y)$ & outcome distribution of the hidden state, Eq.~\eqref{eq:qscore} &
\ref{sec:target}\\
$g_r$; $\hat S$, $\bar S$; $\tau$ & agreement sign; agreement statistic and its expected
value; rejection threshold in standard deviations & \ref{sec:agreetest}\\
$S_1(y)$, $S_2(y)$; $\mathcal W$; $Z_w$ & sums of weights and of squared weights over
the rounds that returned $y$; total weight; weighted collision statistic,
Eq.~\eqref{eq:Zw} & \ref{sec:collisiontest}\\
$\tilde p$; $\chi^2(p\|q)$; $\sigma_Z$ & outcome distribution a server's strategy would
produce on the hidden state; chi-square divergence; honest standard deviation of $Z_w$
& \ref{sec:collisiontest}\\
$R_{\rm agr}$, $R_{\rm col}$ & rounds needed to reject a spoofer, Eq.~\eqref{eq:rounds}
& \ref{sec:collisiontest}\\
$\kappa$ & largest bias of a nonzero linear function of $(s,\sL)$ & \ref{sec:ideal}\\
$\chi$ & bond dimension of a truncated matrix-product simulation & \ref{sec:mps}\\
$H_{\min}$; $\varepsilon(n)$ & min-entropy; soundness error of
Conjecture~\ref{conj:main} & \ref{sec:entropy}\\
\hline\hline
\end{tabular}
\end{center}

% ==================================================================
\section{Proofs of Proposition~\ref{prop:factor}, Theorem~\ref{thm:invariance},
Corollary~\ref{cor:confine} and Lemma~\ref{lem:parity}}\label{app:factor}
% ==================================================================

\begin{proof}[Proof of Proposition~\ref{prop:factor}]
Since $G_t$ acts only on $b$ and the input and output frames agree
outside $b$, write
\begin{equation*}
U_t^\dagger G_tU_{t-1}=I_a\otimes\Gamma.
\end{equation*}
Let $\Lambda_v$ be the diagonal conditional operator on $b$ for
an outside label $v$. Concordance of the output implies that
$\Gamma\Lambda_v\Gamma^\dagger$ is diagonal for every $v$.

A unitary taking these diagonal operators to diagonal operators
maps their common eigenspaces, which are the degenerate sectors,
onto subspaces spanned by computational-basis vectors of the same
dimensions. A permutation $P_t$ of the support labels can therefore
match every input sector with its image. Removing this permutation,
\begin{equation*}
B_t=P_t^\dagger\Gamma,
\end{equation*}
leaves a unitary acting within each sector. Thus $B_t$ is block
diagonal on the sectors and commutes with every $\Lambda_v$, hence
with $\sigma_{t-1}$ when extended by the identity outside $b$.
Since $\Gamma=P_tB_t$, we obtain
\begin{equation*}
G_t=U_tP_tB_tU_{t-1}^\dagger.
\end{equation*}
\end{proof}

\emph{Theorem~\ref{thm:invariance}.} By Proposition~\ref{prop:factor},
$U_t^\dagger G_tU_{t-1}=P_tB_t$, so $\sigma_t=P_tB_t\sigma_{t-1}B_t^\dagger
P_t^\dagger=P_t\sigma_{t-1}P_t^\dagger$, since $[B_t,\sigma_{t-1}]=0$. Iterating from
$\sigma_0=\sum_x\pin(x)\ket x\!\bra x$ gives $\sigma_L=\sum_x\pin(x)\ket{P(x)}\!\bra{P(x)}$.
\hfill$\square$

\emph{Corollary~\ref{cor:confine}.} The frames telescope,
$U_t^\dagger G_t\cdots G_1=P_tB_t\cdots P_1B_1$. Since $B_1$ is block diagonal on the
sectors of $\sigma_0$ on $b_1$ and acts as the identity elsewhere, it maps $\ket x$ into
the span of labels that agree with $x$ off $b_1$ and whose support part lies in the same
sector; such labels have the same probability, so $B_1\ket x$ lies in the span of
$\Lset(x)$. The permutation $P_1$ carries this level set onto the level set
$P_1(\Lset(x))$ of $\sigma_1$, by Theorem~\ref{thm:invariance}. Induction on $t$ gives the
claim. \hfill$\square$

\emph{Lemma~\ref{lem:parity}.} (a) Split a label into its support part $\vec x$ and its
complement $v$. Then $\pin(\vec x,v)=[1+\lambda(-1)^{\langle s|_b,\vec x\rangle}
(-1)^{\langle s|_a,v\rangle}]/2^n$, so the conditional state depends on $\vec x$ only
through the bit $\langle s|_b,\vec x\rangle$, and the two values of that bit give
different conditional states whenever $\lambda\neq0$. The sectors are the two level sets
of $\vec x\mapsto\langle s|_b,\vec x\rangle$, the cosets of its kernel. (b)
$\mathbb E[(-1)^{\langle \omega,x\rangle}]=2^{-n}\sum_x(-1)^{\langle \omega,x\rangle}+\lambda
2^{-n}\sum_x(-1)^{\langle \omega\oplus s,x\rangle}=\delta_{\omega,0}+\lambda\delta_{\omega,s}$. The marginal
on a set $Q$ of qubits is fixed by the coefficients with $\omega$ supported in $Q$; if $Q$
does not contain the support of $s$, all of them vanish for $\omega\neq0$. \hfill$\square$

% ==================================================================
\section{The instance family}\label{app:construction}
% ==================================================================

\begin{enumerate}[leftmargin=*,itemsep=2pt]
\item Draw $s$ uniformly among the vectors of the chosen weight ($|s|\ge3$); set
$s^{(0)}=s$, $\beta_0=0$. The input is diagonal in the computational basis, $U_0=I$.
\item For $t=1,\dots,L$:
\begin{enumerate}[leftmargin=*,itemsep=1pt]
\item draw a support $b_t$ of $\ell=3$ qubits straddling the central cut, with
$s^{(t-1)}|_{b_t}\neq0$;
\item form the two sectors, the cosets of $\ker s^{(t-1)}|_{b_t}$, and order the labels
of each increasingly, so that each sector is read as a two-qubit space;
\item take the Bell map $W_B\ket{jk}=(Z^j\otimes X^k)\ket{\Phi^+}$, set
$W=(I\otimes H)W_B$, and set $B_t=W$ on the sector containing
$\vec 0$ and $W\cdot\mathrm{CZ}$ on the other; the resulting $B_t$ is not Clifford;
\item draw an invertible affine map $\vec x\mapsto K_t\vec x\oplus u_t$ on $\mathbb F_2^3$
uniformly; let $P_t$ be the corresponding permutation, and update the form and flip by
Eq.~\eqref{eq:update};
\item draw Haar-random single-qubit unitaries on $b_t$ to form $U_t$ (equal to $U_{t-1}$
elsewhere) and publish $G_t=U_tP_tB_tU_{t-1}^\dagger$ with $b_t$.
\end{enumerate}
\item Choose $M_q=U_{L,q}\,e^{-i\theta\,\hat n_q\cdot\vec\sigma}$ with random axes,
publish $M$, and keep $U_L$, $\sL$, $\beta$ and $A$.
\end{enumerate}
Optionally, run the LBF on each candidate layer and redraw it if the LBF reads it
without ambiguity (Remark~\ref{rem:reject}).

% ==================================================================
\section{Proof of Theorem~\ref{thm:symtest}}\label{app:symtest}
% ==================================================================

(a) A product state $\bigotimes_j\sigma_j$ with $\mathrm{Tr}[\sigma_jZ]=0$ for every $j$
has every diagonal factor equal to $I/2$, hence is $I/2^n$.

(b) $X_T$ maps $x\mapsto x\oplus T$, an affine permutation and so a Clifford unitary. It
maps $Z_j$ to $\pm Z_j$, so $\mathrm{Tr}[X_T\sigma X_TZ_j]=\pm\mathrm{Tr}[\sigma Z_j]=0$
and the statistic vanishes for every $T$. If every $X_T$ were a symmetry, $p(x\oplus T)=
p(x)$ for all $x,T$, so $p$ would be constant and $\sigma=I/2^n$.

(c) $X_T\sigma_0X_T=[I+\lambda(-1)^{\langle s,T\rangle}Z^s]/2^n$, which equals $\sigma_0$
if and only if $\langle s,T\rangle=0$; since $s\neq0$ this holds for half of the $T$.

(d) For $S:x\mapsto Kx\oplus u$, $S\sigma_0S^\dagger=[I+\lambda(-1)^{\langle s',u\rangle}
Z^{s'}]/2^n$ with $s'=(K^{-1})^{\mathsf T}s$. Its single-qubit $Z$ expectations are
nonzero only if $|s'|=1$, and those of $\sigma_0$ vanish since $|s|\ge2$; so the statistic
vanishes if and only if $|s'|\neq1$. $S$ is a symmetry if and only if $s'=s$ and
$\langle s,u\rangle=0$. As $K$ ranges over $GL(n,\mathbb F_2)$, $s'$ is uniform over the
$2^n-1$ nonzero vectors, $n$ of which have weight one, and half of the offsets $u$ are
orthogonal to $s$; the fraction of test-passing permutations that are symmetries is
$\tfrac12/(2^n-1-n)$. \hfill$\square$

% ==================================================================
\section{Proof of Proposition~\ref{prop:agree}}\label{app:agree}
% ==================================================================

Write $\ket{\varphi_x}=U_L^\dagger C\ket x$. By Corollary~\ref{cor:confine},
$(-1)^\beta Z^{\sL}\ket{\varphi_x}=(-1)^{c(x)}\ket{\varphi_x}$. The outcome is $y$ with
probability $p(y|x)=|\bra y\Omega^\dagger\ket{\varphi_x}|^2$, where $\Omega=\bigotimes_q\Omega_q$. Rounds
are independent, so the $g_r$ are independent $\pm1$ variables, and for uniform $x$
\begin{equation*}
\mathbb E[g]=\frac{(-1)^\beta}{2^n}\,\mathrm{Tr}\bigl[\Delta\,\Omega Z^{\sL}\Omega^\dagger\bigr],
\end{equation*}
with $\Delta=\sum_x(-1)^{c(x)}\ket{\varphi_x}\!\bra{\varphi_x}$.
The map $x\mapsto\ket{\varphi_x}$ restricted to the input level set of class $c$ is a
unitary onto the output level set of class $c$, so $\Delta=(-1)^\beta Z^{\sL}$ and
$\mathbb E[g]=2^{-n}\mathrm{Tr}[Z^{\sL}\Omega Z^{\sL}\Omega^\dagger]$. This factorizes over qubits;
qubits off the support of $\sL$ contribute $1$, and each qubit on it contributes
$\frac12\mathrm{Tr}[(\Omega_q^\dagger Z\Omega_q)Z]=a_q$, because $\Omega_q^\dagger Z\Omega_q$ is traceless
with diagonal $(a_q,-a_q)$. The mean is $A$, and $\lambda$ does not enter. The same
computation gives, for later use,
\begin{equation}
\sum_x\frac{p(y|x)}{2^{n}}=\frac1D,\quad
\sum_x\frac{(-1)^{c(x)}p(y|x)}{2^{n}}=\frac{A\,\xi_y}{D},
\label{eq:marginals}
\end{equation}
with $\xi_y=(-1)^{c_L(y)}$.
If the outcome is independent of the input, the mean factorizes and contains
$\mathbb E[(-1)^{c(x)}]=0$ for uniform $x$. \hfill$\square$

% ==================================================================
\section{Proof of Proposition~\ref{prop:collision}}\label{app:collision}
% ==================================================================

Rounds are independent and identically distributed: $x_r$ uniform, $y_r$ drawn from the
server's response, $w_r=w(x_r)$. Let $h(r,r')=w_rw_{r'}\,\mathbf 1[y_r=y_{r'}]/\plam(y_r)$.
Then $\sum_y[S_1(y)^2-S_2(y)]/\plam(y)=U=\sum_{r\neq r'}h(r,r')$, summed over ordered
pairs, and $Z_w=U/(R-1)-R-2(\mathcal W-R)$.

(a) With $\tilde p(y)=\mathbb E[w_r\mathbf 1[y_r=y]]$, independence gives
$\mathbb E[h]=\sum_y\tilde p(y)^2/\plam(y)$, and $\sum_y\tilde p(y)=\mathbb E[w]=1$. So
$\mathbb E[U]/(R-1)=R\,[\chi^2(\tilde p\|\plam)+1]$, while $\mathbb E[\mathcal W]=R$ for every
server because the client draws $x_r$ uniformly. Hence $\mathbb E[Z_w]=R\,\chi^2(\tilde p
\|\plam)$.

(b) For the honest device $\tilde p=\plam$ by Proposition~\ref{prop:weights} and
Eq.~\eqref{eq:qscore}, so $h_0=\mathbb E[h]=1$ and
$h_1(r)=\mathbb E[h\,|\,r]-h_0=w_r-1$. In the Hoeffding decomposition
$h=h_0+h_1(r)+h_1(r')+h_2(r,r')$, with $\mathbb E[h_2|r]=0$,
\begin{equation*}
U=R(R-1)+2(R-1)(\mathcal W-R)+\sum_{r\neq r'}h_2(r,r'),
\end{equation*}
so $Z_w=\sum_{r\neq r'}h_2(r,r')/(R-1)$, a degenerate $U$-statistic. Its terms are
uncorrelated unless they involve the same unordered pair, so
$\mathrm{Var}(Z_w)=2R\,\mathbb E[h_2^2]/(R-1)$ with
$\mathbb E[h_2^2]=\mathbb E[h^2]-h_0^2-2\mathbb E[h_1^2]=\mathbb E[h^2]-1-2\lambda^2$,
since $\mathrm{Var}(w)=\lambda^2$. Finally $\mathbb E[h^2]=\sum_ym_2(y)^2/\plam(y)^2$ with
$m_2(y)=\mathbb E[w_r^2\mathbf 1[y_r=y]]=\sum_x2^{-n}w(x)^2p(y|x)$. Using
$w^2=1+\lambda^2+2\lambda(-1)^{c(x)}$ and Eq.~\eqref{eq:marginals},
$m_2(y)=[1+\lambda^2+2\lambda A\xi_y]/D$ and $\plam(y)=[1+\lambda A\xi_y]/D$, with
$\xi_y=\pm1$ on half of the strings each. Summing gives Eq.~\eqref{eq:sigma}. \hfill$\square$

% ==================================================================
\section{Proof of Theorem~\ref{thm:poly}}\label{app:poly}
% ==================================================================

(i) The client keeps $s$, $\sL$, $\beta$, $A$, the frames, and the $L$ tables of $d_b$
entries. (ii) On layer $t$ it computes the restricted form and its cosets, $O(\ell\,d_b)$;
draws an affine map, $O(\ell^2)$; updates the form,
$O(n\ell)$ at most; and forms the block factor and the compiled gate, $O(d_b^{\,3})$.
(iii) By Lemma~\ref{lem:parity}(a) the cosets are the exact sector partition, so
$[B_t,\sigma_{t-1}]=0$ holds exactly. (iv) The weight and the agreement sign need one
parity each; $\plam(y)$ needs one parity and $A$. (v) The frames are local, so
$U_t^\dagger G_tU_{t-1}=P_tB_t$ does not involve them; a product of single-qubit
unitaries acts separately on the two sides of any bipartition, so it preserves Schmidt
coefficients and commutes with truncation. \hfill$\square$

% ==================================================================
\section{Proof of Theorem~\ref{thm:ideal}}\label{app:ideal}
% ==================================================================

Write the composite affine permutation as $P(x)=Kx\oplus u$. Class preservation,
$c_L(P(x))=c(x)$ for all $x$, gives $K^{\mathsf T}\sL=s$ and $\beta=\langle\sL,u\rangle$.
Define
\begin{equation}
\kappa=\max_{(e,e')\neq(0,0)}\Bigl|\mathbb E\bigl[(-1)^{\langle s,e\rangle\oplus
\langle\sL,e'\rangle}\bigr]\Bigr| ,
\label{eq:kappa}
\end{equation}
the expectation taken over the client's sampling. For $r\neq r'$,
\begin{equation*}
g_rg_{r'}=(-1)^{\langle s,\,x_r\oplus x_{r'}\rangle\oplus\langle\sL,\,y_r\oplus y_{r'}
\rangle},
\end{equation*}
in which $\beta$ has cancelled. The labels are drawn by the client and the outcomes are
computed from the server's view and private randomness, both independent of
$(s,K,u,U_L)$. Conditioning on them fixes $(e,e')=(x_r\oplus x_{r'},y_r\oplus y_{r'})$.
If $(e,e')\neq(0,0)$, $|\mathbb E[g_rg_{r'}]|\le\kappa$ by Eq.~\eqref{eq:kappa}. If
$(e,e')=(0,0)$ the product is $1$, which requires $x_r=x_{r'}$, an event of probability
$2^{-n}$. Hence
\begin{equation*}
\mathbb E[\hat S^2]=\frac1{R^2}\Bigl[R+\sum_{r\neq r'}\mathbb E[g_rg_{r'}]\Bigr]
\le\frac1R+2^{-n}+\kappa .
\end{equation*}
Acceptance requires $\hat S\ge A-\delta_R>0$, hence $\hat S^2\ge(A-\delta_R)^2$, and
Markov's inequality gives the bound. For the estimate of $\kappa$: if $e'=0$ and $e\neq0$,
$|\mathbb E_s(-1)^{\langle s,e\rangle}|\le\nu_n$, because the sum of
$(-1)^{\langle s,e\rangle}$ over all $s$ vanishes and the excluded vectors of weight at
most two number $1+n+n(n-1)/2$. If $e'\neq0$, $\langle\sL,e'\rangle=\langle s,K^{-1}e'
\rangle$, so the exponent is $\langle s,e\oplus K^{-1}e'\rangle$, whose bias is at most
$\Pr[K^{-1}e'=e]+\nu_n\le\mu+\nu_n$ when $K$ is independent of $s$, where
$\mu=\max_{e'\neq0,e}\Pr[K^{-1}e'=e]$, $\nu_n=(1+n+n(n-1)/2)/N_s$, and $N_s$ is the
number of vectors of weight at least three. For a well-mixing composite permutation
$\mu$ is of order $2^{-n}$. \hfill$\square$

% ==================================================================
\section{Proof of Proposition~\ref{prop:trapdoor}}\label{app:trapdoor}
% ==================================================================

\emph{Reading the circuit.} The sectors on $b_1$ follow from $s$ by
Lemma~\ref{lem:parity}(a), and the input frame is the computational basis. Given the sectors on $b_t$ and the frame
$U_{t-1}$, Proposition~\ref{prop:known} gives a frame $U_t$ and the image of each sector
under the gate, in time $O(d_b^{\,3})$. The image sets are the cosets of the new
restricted form, and the image of the class-$0$ sector carries class $0$; this fixes
$s^{(t)}|_{b_t}$ and $\beta_t$ uniquely, whatever relabeling inside a sector the reading
contains. Induction gives $\sL$, $\beta$ and $U_L$. On qubits off the support of $\sL$
the frame does not enter $A$; on qubits in it, the frame is fixed up to a phase and an
exchange of basis vectors, which changes the sign of $a_q$ and of the corresponding
class bit together, so $A(-1)^\beta$ is recovered.

\emph{Answering.} Given $x$, the algorithm sets $c'=c(x)$ with probability $(1+A)/2$ and
$c'=1-c(x)$ otherwise, and returns $y$ uniform among the labels with
$c_L(y)=c'$. Then $g_r=(-1)^{c(x_r)+c'}$ are independent $\pm1$
variables of mean $A$, exactly as for the honest device (Proposition~\ref{prop:agree}).
The answer distribution is $p'(y|x)=[1+A(-1)^{c(x)+c_L(y)}]/D$,
so $\tilde p(y)=\sum_x2^{-n}w(x)p'(y|x)=\plam(y)$ and
$m_2(y)=\sum_x2^{-n}w(x)^2p'(y|x)=[1+\lambda^2+2\lambda A\xi_y]/D$, as for the honest
device. The proof of Proposition~\ref{prop:collision}(b) uses only $\tilde p=\plam$, the
distribution of the weights and $m_2$, so $Z_w$ has the honest mean and variance.
\hfill$\square$

% ==================================================================
\section{Proof of Proposition~\ref{prop:soundentropy}}\label{app:entropy}
% ==================================================================

Let $\hat Y$ be the adversary's guess and $k$ the probability that $\hat Y=Y$ conditioned
on acceptance, so that $\Pr[\hat Y=Y\wedge\mathrm{acc}]=k\Pr[\mathrm{acc}]$. Consider the
classical algorithm that draws fresh shared randomness, computes $E$ and then $\hat Y$
as the adversary would, and submits $\hat Y$. Whether a submission passes depends only on
the submitted labels and on the client's inputs and weights, and the pair formed by
$\hat Y$ and those inputs has the same distribution whether the randomness is fresh or
shared. The algorithm therefore passes with probability at least $k\Pr[\mathrm{acc}]$,
which is at most $\varepsilon(n)$ by hypothesis. Taking the best predictor,
$H_{\min}(Y|E,\mathrm{acc})=-\log_2\max k$, and the bound follows. \hfill$\square$

% ==================================================================
\section{The local-basis finder diagnostic}\label{app:lbf}
% ==================================================================

For a layer with support $b$ and input frame $U_{t-1}$, the diagnostic follows
Ref.~\cite{Cable_2015}. A candidate is a set $X$ of support labels with $1\le|X|\le
d_b/2$, $162$ candidates at $\ell=3$. For each candidate it forms
$O_X=G_tU_{t-1}\Pi_XU_{t-1}^\dagger G_t^\dagger$ and, for every qubit $q$ of the support,
solves the linear equations $[O_X,\pi_q\otimes I]=0$ for a $2\times2$ operator $\pi_q$. The
candidate is valid if every qubit either has a two-dimensional solution space, which
fixes its basis, or is unconstrained, and if the resulting product basis diagonalizes
$O_X$. Two valid candidates are incompatible if they fix different bases on some qubit,
and the layer halts the LBF if any incompatible pair exists. A rank-one candidate
$X=\{\vec x\}$ is valid exactly when the column $G_tU_{t-1}\ket{\vec x}$ is a product state; the
rank-one count of Sec.~\ref{sec:instances} is the number of distinct tuples of
single-qubit purities among the $d_b$ columns. All measurements were made with the
scripts accompanying this paper.

\end{document}